\documentclass{amsart}
\usepackage{amsfonts}
\usepackage{latexsym}
\usepackage{amscd}
\usepackage{amsmath}
\usepackage{amssymb}
\usepackage{graphicx}
\usepackage[mathscr]{eucal}

\newtheorem{theo}{Theorem}[section]
\newtheorem{prop}[theo]{Proposition}

\newtheorem{de}[theo]{Definition}

\newtheorem{example}[theo]{Examples}

\newcommand{\gp}{\mathbb{P}}

\title[$\delta$-sequences and codes defined by
valuations]{$\delta$-Sequences and Evaluation Codes defined by Plane
Valuations at Infinity}
\author{C.~Galindo \and F.~Monserrat}
\curraddr{C. Galindo: Departament de Matem\`{a}tiques, Universitat
Jaume I, Campus de Riu Sec. s/n, 12071 Castell\'{o} (Spain); F.
Monserrat: Instituto Universitario de Matem\'atica Pura y Aplicada,
Universidad Polit\'ecnica de Valencia, Camino de Vera s/n, 46022
Valencia (Spain) } \email{ galindo@mat.uji.es \hskip 0.3cm
framonde@mat.upv.es}

\thanks{Supported by Spain Ministry of Education
 MTM2007-64704, JCyL VA025A07 and Bancaixa P1-1A2005-08.
We thank C. Munuera and F. Torres for helpful discussions.}

\begin{document}
\maketitle

\begin{abstract}
We introduce the concept of $\delta$-sequence. A $\delta$-sequence
$\Delta$ generates a well-ordered semigroup $S$ in $\mathbb{Z}^2$
or $\mathbb{R}$. We show how to construct (and compute parameters)
for the dual code of any evaluation code associated with a weight
function defined by $\Delta$ from the polynomial ring in two
indeterminates to a semigroup $S$ as above. We prove that this is
a simple procedure which can be understood by considering a
particular class of valuations of function fields of surfaces,
called plane valuations at infinity. We also give algorithms to
construct an unlimited number of $\delta$-sequences of the
different existing types, and so this paper provides the tools to
know and use a new large set of codes.
\end{abstract}

\section{Introduction}

BCH and Reed-Solomon codes can be decoded from the sixties by using
the Berlekamp-Massey algorithm \cite{3gp,20gp}. A paper by Sakata
\cite{35gp} allowed to derive the Berlekamp-Massey-Sakata algorithm
which can be used to efficiently decode certain Algebraic Geometry
codes on curves \cite{51hoh}. Indeed, this last algorithm allows us
to get fast implementations of the modified algorithm of
\cite{36hoh,sur100} (see \cite{sur49,sur46}) and of the majority
voting scheme for unknown syndromes of Feng and Rao \cite{sur23}
(see \cite{sur49,sur65,sur89,sur90}). For a survey on the decoding
of Algebraic Geometry codes one can see \cite{este}.

The so-called order functions were introduced in \cite{h-l-p}, which
in this initial stage had as image set a sub-semigroup of the set of
nonnegative integers. An order function defines a filtration of
vector spaces contained in its definition domain which, together
with an evaluation map, provide two families of error correcting
codes (evaluation codes and their duals). We note that the one-point
geometric Goppa codes or weighted Reed-Muller codes can be regarded
as codes given by order functions.

A type of particularly useful order functions are the weight
functions. Goppa  distance and Feng-Rao distances (also called order
bounds) are lower bounds for the minimum distance of their
associated  dual codes which can be decoded by using the mentioned
Berlekamp-Massey-Sakata algorithm, correcting a number of errors
that depends on the above bounds \cite{h-l-p,sul}. Furthermore,
Matsumoto in \cite{miu} proved that their associated order domains
are affine coordinate rings of algebraic curves with exactly one
place at infinity.

Recently in \cite{g-p}, the concept of order function (and the
related ones of weight function and order domain) have been
enlarged by admitting that the image of those functions can be a
well-ordered semigroup. Order domains are close to Groebner
algebras and they allow to use the theory of Groebner basis
\cite{2gp}.  This enlargement provides a greater variety of
evaluation codes and it has the same advantages (bounds of minimum
distance and fast decoding) that we had with the first defined
concept.

Weight functions were introduced to give an elementary treatment to
some Algebraic Geometry codes. Nevertheless, a deeper study of these
functions with the help of Algebraic Geometry may derive in
obtaining new good linear codes. This is the line of this paper. The
main notion we introduce is an extension to elements in
$\mathbb{Z}^2$, $\mathbb{Q}$ and $\mathbb{R}$ of the classical
concept of $\delta$-sequence defined by the Abhyankar-Moh
conditions. These $\delta$-sequences provide valuations in function
fields of surfaces (related to curves with only one place at
infinity) and, for that reason, weight functions. This allows us to
focus our development to an application in Coding Theory that
consists of studying the evaluation codes given by those weight
functions.

Valuations and  weight functions are very close objects  as one can
see in \cite{sul}. Apart from the simpler case of curves, it is only
available a classification of valuations of function fields of
nonsingular surfaces (also called plane valuations)
\cite{zar,spi,kiy}, that allows us to decide which of them are
suitable for providing, in an explicit manner, order domains and
evaluation codes. The first examples that use that classification
were given in \cite{sul} and a more systematic development can be
found in \cite{ga-sa}. Although both papers have the same
background, they provide different types of examples and it seems
that examples in \cite{sul} cannot be obtained from the development
in \cite{ga-sa}. In this paper we give theoretic foundations to
provide families of weight functions that contain as particular
cases the examples concerning plane valuations given in \cite{sul}.
This leads us  to a deeper knowledge of the involved evaluation
codes what allows us to get explicitly a large set of codes (and
associated parameters) from a simple input which can be easily
determined.

More explicitly, in \cite{ga-sa} it is assumed that $R$ is a
2-dimensional Noetherian regular local ring with quotient field $K$
and that $R$ has an algebraically closed coefficient field $k$ of
arbitrary characteristic. By picking a plane valuation $\nu$ (of $K$
centered at $R$) belonging to  any type of the above mentioned
classification, except the so called divisorial valuations, it is
found an order domain $D$ attached to the weight function $- \nu$,
whose image semigroup is the value semigroup $S:= \{ f \in R
\setminus \{0\} | \nu (f) \geq 0 \}$ of $\nu$. Moreover, in that
paper parametric equations and examples of those weight functions,
and also bounds for the minimum distance of the corresponding dual
evaluation codes are given. In this paper, we shall consider a
different point of view. This is to look for weight functions, whose
order domain is the polynomial ring in two indeterminates $T :=
k[x,y]$, derived from valuations of the quotient field $k(x,y)$ of
$T$, where, now, $k$ needs not be algebraically closed. The
mentioned paper \cite{sul} provides some examples of this type with
monomial and non-monomial associated ordering, but no additional
explanation is supplied.

We shall show that the key to get this last type of weight functions
(with value semigroup generated by the so-called $\delta$-sequences)
is to use for their construction certain class of plane valuations,
that we shall name plane valuations at infinity (see Definition
\ref{preaprva} and Proposition \ref{peso}). These valuations are
those given by certain families of infinitely many plane curves, all
of them having only one place at infinity. This constitutes, in some
sense, an extension  of the above result by Matsumoto.

Our development allows us to reach our main goal, which consists
of explicitly constructing weight functions from $T$, attached to
plane valuations at infinity, for which it is easy to determine
its image semigroup and the filtration of vector spaces in $T$ we
need to construct the codes, and either to bound or to determine
the corresponding Feng-Rao and Goppa distances (see Theorems
\ref{main} and \ref{fr} and Proposition \ref{go}). The unique
input to compute those data is the so-called $\delta$-sequence of
the weight function which also determines it. It is also worth
adding that the obtained image semigroups have a behavior close to
the one of telescopic semigroups and that the so-called
approximates of the valuation at infinity (see Definition
\ref{aprva}) make easy the computation of the mentioned
filtration.

The implementation of these new codes is very simple since to
compute the data to use them is straightforward from the
$\delta$-sequences that define the weight functions; this paper
provides algorithms (if necessary) to get an unlimited number of
$\delta$-sequences of any existing type. In fact, to provide a
family of codes $\{E_\alpha\}_{\alpha \in S}$ and its dual family
$\{C_\alpha\}_{\alpha \in S}$ with alphabet code $k$ ($k$ is any
finite field), we only need to pick a $\delta$-sequence $\Delta$
(Definition \ref{buena}) and a set $ \mathfrak{A}$ of $n$ different
points in $k^2$. $\Delta$ can be finite (included in $\mathbb{Z}^2$
or $\mathbb{R}$) or infinite (included in $\mathbb{Q}$); in the last
case, we need only a finite subset of $\Delta$ that depends on $k$
and $ \mathfrak{A}$. To obtain $\delta$-sequences is easy following
the procedure given in Section \ref{construction}. Moreover $S$ is
the semigroup spanned by $\Delta$ and must be ordered
lexicographically when it is in $\mathbb{Z}^2$ and usually ordered
otherwise. $E_\alpha$ is the image $ev(O_\alpha)$ of a vector space
$O_\alpha$ of $k[x,y]$, where $ev$ is the map $ev:k[x,y] \rightarrow
k^n$ that evaluates polynomials at the points in $ \mathfrak{A}$.
$C_\alpha$ is the dual vector space of $E_\alpha$. Finally, bases
for the spaces $O_\alpha$ can be computed without difficulty since
they are formed by polynomials as in (\ref{misq}) (see the beginning
of Section \ref{construction}), where the exponents are obtained
from $\Delta$ by an algorithm similar to Euclid's one (see the
beginning of Section \ref{42}).

Next, we briefly describe the content of the paper. Valuations on
function fields of surfaces and their classification are fundamental
for understanding the development of  this paper and Section
\ref{tres} is devoted to provide a short summary  for it. We only
use a particular subset of these valuations, plane valuations at
infinity. Although the definition is given in Section \ref{cuatro},
to manipulate these valuations we need some knowledge about plane
curves with only one place at infinity and this information is
displayed in Section \ref{dos}. The development of this section
shortens Section \ref{tres} due to the closeness between them.
Section \ref{cuatro} contains the core of the paper: we give the
definition of $\delta$-sequence and algorithms for constructing
$\delta$-sequences. We also explain the way to use them to get
weight functions associated with valuations at infinity and how to
compute bases of the vector spaces filtration needed to obtain the
corresponding codes. There are $\delta$-sequences giving rise to
valuations at infinity only for three of the five types of the above
mentioned classification for plane valuations. These
$\delta$-sequences provide weight functions with monomial and
non-monomial corresponding orderings. Some examples are also
supplied in this section. Finally, Section \ref{cinco} studies the
Feng-Rao and  Goppa  distances for dual codes given by
$\delta$-sequences (see Theorem \ref{fr} and Proposition \ref{go}).
A useful property to study these distances is that
$\delta$-sequences generate semigroups similar to telescopic ones,
which we call generalized telescopic semigroups. Moreover, the
$\delta$-sequences which are in ${\mathbb Z}^2$ span simplicial
semigroups and thus we can compute the least Feng-Rao distance using
an algorithm by Ruano \cite{rua}. We end this paper proving, in
Proposition \ref{rs}, that Reed-Solomon codes are a particular case
of codes given by $\delta$-sequences included in $\mathbb{Z}^2$ and
giving several examples of dual codes (showing their parameters)
defined by $\delta$-sequences of all mentioned types.

\section{Plane curves with only one place at infinity}
\label{dos}

We devote this section to summarize several known results concerning
plane curves with only one place at infinity because this geometric
concept supports the definition of those weight functions that will
be useful for our purposes. References for the subject are
\cite{pa2,pa1,pa50,pa45,pa53,pa21}. Some of the results concerning
this type of curves have been used by Campillo and Farr\'an in
\cite{far-cam} for computing the Weierstrass semigroup and the least
Feng-Rao distance of the corresponding Goppa codes attached to
singular plane models for curves with only one place at infinity. We
begin with the definition of a key concept for this paper. Denote by
$\mathbb{N}_{> (\geq) 0}$ the set of positive (nonnegative)
integers.

\begin{de}
\label{delta} {\rm A {\em $\delta$-sequence in  $\mathbb{N}_{>
0}$} is a finite sequence of positive integers $\Delta = \{
\delta_i \}_{i=0}^g$, $g \geq 0$, satisfying the following three
conditions
\begin{itemize}
\item[(1)] If $d_i=\gcd (\delta_0,\delta_1,\ldots,\delta_{i-1})$,
for $1\leq i\leq g+1$, and $n_i=d_i/d_{i+1}$, $1\leq i\leq g$,
then $d_{g+1}=1$ and $n_i>1$ for $1\leq i\leq g$.

\item[(2)] For $1\leq i\leq g$, $n_i\delta_i$ belongs to the
semigroup generated by $\delta_0,\delta_1,\ldots, \delta_{i-1}$,
that we usually denote $\langle
\delta_0,\delta_1,\ldots,\delta_{i-1}\rangle$.

\item[(3)] $\delta_0>\delta_1$ and $\delta_i<\delta_{i-1} n_{i-1}$
for $i=2,3,\ldots,g$.
\end{itemize}}
\end{de}

Above conditions are usually called Abhyankar-Moh conditions. We
denote by $S_\Delta$ the sub-semigroup of $\mathbb{N}_{\geq 0}$
spanned by $\Delta$.

Along this paper $\mathbb{P}_k^2$ (or $\mathbb{P}^2$ for short)
stands for the projective plane over a field $k$ of arbitrary
characteristic. Now, let us state the definition of plane curve with
only one place at infinity.

\begin{de}
\label{de1} {\rm Let $L$ be the line at infinity in the
compactification of the affine plane to $\gp^2$. Let $C$ be a
projective absolutely irreducible curve of $\gp^2$ (i.e.,
irreducible as a curve in $\gp^2_{\overline{k}}$, $\overline{k}$
being the algebraic closure of $k$). We shall say that $C$ {\it has
only one place at infinity} if the intersection $C\cap L$ is a
single point $p$ (the one at infinity) and $C$ has only one branch
at $p$ which is rational (that is, defined over $k$).}
\end{de}

Set $C$ a curve with only one place at infinity. Denote by $K$ the
quotient field of the local ring $\mathcal{O}_{C,p}$; the germ of
$C$ at $p$ defines a discrete valuation on $K$ that we set
$\nu_{C,p}$, which allows us to state the following

\begin{de}
{\rm Let $C$ be a curve with only one place at infinity given by
$p$. The {\it semigroup at infinity} of $C$ is the following
sub-semigroup of $\mathbb{N}_{\geq 0}$:
\[
S_{C,\infty}:= \left \{ - \nu_{C,p} (h) | h \in T \right \},
\]
where $T$ is the $k$-algebra $\mathcal{O}_{C}( C \setminus
\{p\})$.}
\end{de}

For any curve $C$ with only one place at infinity, except when the
characteristic of $k$ divides the degree of $C$, it can be proved
that there is a $\delta$-sequence in $\mathbb{N}_{>0}$, $\Delta$,
such that $S_{C,\infty} = S_\Delta$ \cite{pa2}. Conversely, as we
shall precise later, for any field $k$ and any $\delta$-sequence in
$\mathbb{N}_{>0}$, there exists a plane curve $C$ with only one
place at infinity such that $S_\Delta =S_{C,\infty} $ and $\delta_0$
is the degree of $C$.

$C$ has a singularity at $p$, except when $g=0$. Consider the
infinite sequence of morphisms
\begin{equation}\label{infiniteseq}
\cdots \rightarrow X_{i+1}\rightarrow X_i\rightarrow \cdots
\rightarrow X_1 \rightarrow X_0:=\gp^2,
\end{equation}
where $X_1\rightarrow X_0$ is the blowing-up at $p_0:=p$ (the point
at infinity) and, for each $i\geq 1$, $X_{i+1}\rightarrow X_i$
denotes the blowing-up of $X_i$ at the unique point $p_i$ which lies
on the strict transform of $C$ and the exceptional divisor created
by the preceding blowing-up; notice that $p_i$ is defined over $k$,
since the branch of $C$ at $p$ is rational. It is well-known that
there exists a minimum integer $n$ such that, if $\pi : X_n
\rightarrow \gp^2$ denotes the composition of the first $n$
blowing-ups, the germ of the strict transform of $C$ by $\pi$ at
$p_n$ becomes regular and transversal to the exceptional divisor.
This gives the (minimal embedded) resolution of the germ of $C$ at
$p$. The essential information, the (topological) equisingularity
class of the germ, can be given in terms of its sequence of {\it
Newton polygons} \cite[III.4]{cam} or by means of its {\it dual
graph} (see \cite{d-g-n} within a more general setting or
\cite{pa21} for a slightly different version). This information
basically provides the number and the position of blowing-up centers
of $\pi$: these can be placed either on a free point (not an
intersection of two exceptional divisors) or on a satellite point.
In this last case, it is also important to know whether, or not, the
blowing-up center belongs to the last but one created exceptional
divisor. Thinking of blowing-up centers, we shall say that a center
$p_i$ is proximate to other $p_j$ whenever $p_i$ is on any strict
transform of the divisor created after blowing-up at $p_j$.

The complete information of the resolution process can be described
by means of the so-called Hamburger-Noether expansion (HNE for
short) \cite{cam,c-f}, being this expansion specially useful when
the characteristic of the field $k$ is positive. More explicitly,
let $\{u',v'\}$ be local coordinates of the local ring
$\mathcal{O}_{C,p}$; in those coordinates the HNE of $C$ at $p$ has
the form
\begin{equation}\label{hne}
\begin{array}{lccl}
&v' & = & a_{01}u'+a_{02}u'^{2}+ \cdots + a_{0h_{0}}u'^{h_{0}}+u'^{h_{0}}w_{1} \\
&u' & = & w^{h_{1}}_{1}w_{2} \\
&\vdots & \nonumber & \vdots \\
&w_{s_{1}- 2} & = & w^{h_{s_{1}-1}}_{s_{1}-1}w_{s_{1}}\\
&w_{s_{1}-1} & = & a_{s_{1}k_{1}}w^{k_{1}}_{s_{1}} + \cdots
 +a_{s_{1}h_{s_{1}}}w^{h_{s_{1}}}_{s_{1}}+w^{h_{s_{1}}}_{s_{1}}w_{s_{1}+1}
 \\
\,& \vdots & \nonumber & \vdots \\
&w_{s_{g}- 1} & = & a_{s_{g}k_{g}}w^{k_{g}}_{s_{g}} + \cdots
\end{array}
\end{equation}
where the family $\{s_i\}_{i=0}^{g}$, $s_0=0$,  of nonnegative
integers is the set of indices corresponding to the free rows of the
expression, that is those rows that express the blowing-ups at free
points (they are those that have some nonzero $a_{jl} \in k$) and
the main goal (of the HNE) is that it gives local coordinates of the
transform of the germ of $C$ at $p$ in each center of blowing-up.
The local coordinates after $\{u',v'\}$ are $\{u', (v'/u')-
a_{01}\}$ and so on.

The dual graph $\Gamma$ associated with the above germ of curve is
a tree such that each vertex represents an exceptional divisor of
the sequence $\pi$ and two vertices are joined by an edge whenever
the corresponding divisors intersect. Additionally, we label each
vertex with the minimal number of blowing-ups needed to create its
corresponding exceptional divisor. The dual graph can be done by
gluing by their vertices $st_i$ subgraphs $ \Gamma_i $ $ (1 \leq i
\leq g) $ corresponding to  blocks of data $ B_i=
\{h_{s_{i-1}}-k_{i-1} +1, h_{s_{i-1}+1}, h_{s_{i-1}+2}, \ldots,
h_{s_{i}-1}, k_{i} \} $ (with $k_0=0$), which represent the
divisors involved in the part of HNE of the germ between two free
rows. That is $\Gamma_i$ contains divisors corresponding to
$h_{s_{i-1}}-k_{i-1} +1$ free points and to sets of $h_j$
($s_{i-1}+1 \leq j \leq s_{i}-1$) and $k_i$ proximate points to
satellite ones. Each subgraph $\Gamma_i$ starts in the vertex
$st_{i-1}$ and ends in $st_i$ containing, among others, the vertex
$\rho_i$. So, the dual graph has the shape depicted in Figure
\ref{fig0}.

\begin{figure}[h]
\[
\unitlength=1.00mm
\begin{picture}(80.00,30.00)(-10,3)
\thicklines \put(-5,30){\line(1,0){41}}
\put(44,30){\line(1,0){16}} \put(38,30){\circle*{0.5}}
\put(40,30){\circle*{0.5}} \put(42,30){\circle*{0.5}}
\put(30,10){\line(0,1){20}} \put(50,20){\line(0,1){10}}
\put(60,0){\line(0,1){30}} \put(10,15){\line(0,1){15}} \thinlines
\put(20,30){\circle*{1}} \put(30,30){\circle*{1}}
\put(50,30){\circle*{1}} \put(60,30){\circle*{1}}

\put(30,20){\circle*{1}} \put(60,20){\circle*{1}}
\put(60,10){\circle*{1}} \put(10,30){\circle*{1}}
\put(30,10){\circle*{1}} \put(50,20){\circle*{1}}
\put(60,0){\circle*{1}} \put(-5,30){\circle*{1}}
\put(0,30){\circle*{1}} \put(5,30){\circle*{1}}
\put(15,30){\circle*{1}} \put(25,30){\circle*{1}}
\put(35,30){\circle*{1}} \put(45,30){\circle*{1}}
\put(55,30){\circle*{1}} \put(10,25){\circle*{1}}
\put(10,20){\circle*{1}} \put(10,15){\circle*{1}}
\put(30,25){\circle*{1}} \put(30,15){\circle*{1}}
\put(35,30){\circle*{1}} \put(-9,25){{\bf 1}=$\rho_0$}
\put(11.5,14){$\rho_1$} \put(4.5,19){$\Gamma_1$}
\put(31.5,9){$\rho_2$} \put(24.5,14){$\Gamma_2$}
\put(61.5,-1){$\rho_g$} \put(54.5,4){$\Gamma_g$}
\put(9,32){$st_1$} \put(29,32){$st_2$} \put(57.5,32){$st_g$}
\end{picture}
\]
\caption{The dual graph of a germ of curve} \label{fig0}
\end{figure}
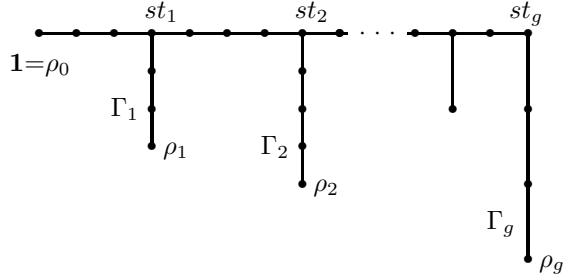

Set $E_{s_i}$ $(1 \leq i \leq g)$ the exceptional divisor obtained
after blowing-up the last free point corresponding to the subgraph
$\Gamma_i$. It corresponds to the vertex $\rho_i$ in the dual
graph. An irreducible germ of curve  at $p$, $\psi$, is said to
{\it have maximal contact of genus $i$} with the germ of $C$ at
$p$, if the strict transform of $\psi$ in the (corresponding germ
of the) surface containing $E_{s_i}$ is not singular and meets
transversely $E_{s_i}$  and no other exceptional curves.

For convenience, throughout this paper we fix homogeneous
coordinates $(X:Y:Z)$ on $\mathbb{P}^2$. $Z=0$ will be the line at
infinity and $p=(1:0:0)$. Set $(x,y)$ coordinates in the chart
$Z\neq 0$ and $(u=y/x, v=1/x)$ coordinates around  the point at
infinity. We shall assume that the curve $C$ is defined by a monic
polynomial $f(x,y)$ in the indeterminate $y$ with coefficients in
$k[x]$.

The so-called approximate roots of $C$ \cite{pa1} are an important
tool to get a $\delta$-sequence $\Delta$ in $\mathbb{N}_{> 0}$ such
that (for suitable characteristic of $k$) $S_\Delta = S_{C,\infty}$.
One can see in \cite{far-cam} an algorithm for computing them (see
also \cite{pa21} for the complex case). We need not use approximate
roots for the development of this paper but a close weaker concept
which is given in the following definition (see \cite{pa46}). Its
main advantage (explicit description for certain explicit curves
obtained only from the $\delta$-sequence in $\mathbb{N}_{> 0}$) is
showed in Proposition \ref{siete}.

\begin{de}\label{z0}
\label{aprox} {\rm Assuming the above notation, a sequence of
polynomials in $k[x,y]$}
$$q_0^*(x,y), q_1^*(x,y),  \ldots, q_g^*(x,y)$$ {\rm is a {\it family
of approximates} for the curve $C$ given by $f(x,y)$ if the
following conditions hold:
\begin{enumerate}
    \item $q_0^* (x,y)=x$, $q_1^* (x,y)=y$, $\delta_0^* :=- \nu_{C,p} (q_0^*)= \deg_y
    (f)$ and $\delta_1^* :=- \nu_{C,p} (q_1^*)$.
    \item $q_i^*(x,y)$ ($1 < i \leq g$) has degree
    $\delta_0^* /d_i$ and it is monic in the indeterminate $y$, where\\
    $d_i=\gcd(\delta_0^*, \delta_1^*, \ldots, \delta_{i-1}^*)$,
     being $\delta_i^* := - \nu_{C,p} (q_i^*)$.
    \item The germ of curve at $p$ given by the local expression
    of $q_i^* (x,y)$ $(1 < i \leq g)$ in the coordinates $(u,v)$ has maximal contact with the germ of $C$ at
    $p$, of genus $i$ when $\delta^*_0-\delta^*_1$ does not divide
    $\delta^*_0$ and of genus $i-1$ otherwise.

\end{enumerate}
}
\end{de}
By an abuse of notation, when we set $- \nu_{C,p} (q_i^*)$,
$q_i^*$ stands for the element in the fraction field of
$\mathcal{O}_{C,p} $ that it defines. On the other hand, under the
conditions of Abhyankar-Moh Theorem, that is, the characteristic
of $k$ does not divide the degree of the curve $C$, approximate
roots are a family of approximates for $C$.

Now, let $\Delta = \{\delta_i\}_{i=0}^g$ be a $\delta$-sequence in
$\mathbb{N}_{> 0}$. It is well-known the existence of a unique
expression of the form

\begin{equation}
n_i \delta_i = \sum_{j=0}^{i-1} a_{ij} \delta_j,
\end{equation}
where $a_{i0} \geq 0$ and $0 \leq a_{ij} < n_j$, for $1 \leq j
\leq i-1$. Set $q_0:=x$ $q_1:=y$ and, for $1 \leq i \leq g$,
\begin{equation}
 \label{misq} q_{i+1}:= q_{i}^{n_{i}} - t_i  \prod_{j=0}^{i-1} q_j^{a_{ij}},
\end{equation}
where $t_i \in k \setminus \{0\}$ are arbitrary. Although all the
results in this paper concerning these polynomials hold for any
family of parameters $\{t_i\}_{i=1}^g$, we fix for convenience
$t_i=1$ for all $i$. Then, by applying the algorithms relative to
Newton polygons of a germ of curve given by Campillo in
\cite[III.4]{cam} to the germ given by $q_{g+1}$, it holds the
following result (see \cite[Section 4]{pa46} for more details),
where we notice that there is no restriction for the characteristic
of the field $k$.

\begin{prop}
\label{siete} The equality $q_{g+1}=0$ defines a plane curve $C$
with only one place at infinity such that $S_{C,\infty} =
S_{\Delta}$ and the set $\{q_i\}_{i=0}^{g}$ is a family of a
approximates for $C$ such that $- \nu_{C,p} (q_i)=\delta_i$ for all
$i=0,1,\ldots,g$.
\end{prop}

The sequence of Newton polygons and the dual graph of the germ of a
curve with only one place at infinity can be recovered from a
$\delta$-sequence in $\mathbb{N}_{> 0}$, $\Delta= \{ \delta_0,
\delta_1, \ldots, \delta_s \}$, associated with it. We assume that
the Newton polygons are given by segments $P_i$ $(0 \leq i \leq
g-1)$ joining the points $(0,e_i)$ and $(m_i,0)$, $e_i,m_i \in
\mathbb{N}_{> 0}$. If $\delta_0-\delta_1$ does not divide $\delta_0$
then $s=g$ and
\begin{equation*}
e_0=\delta_0-\delta_1, \;\;\; e_i=d_{i+1}
\end{equation*}
\begin{equation*}
m_0=\delta_0, \;\;\; m_i=n_i\delta_i-\delta_{i+1}
\end{equation*}
for $1\leq i\leq s-1$.  Otherwise, $s=g+1$ and
\begin{equation*}
e_0=d_2=\delta_0-\delta_1, \;\;\; e_i=d_{i+2}
\end{equation*}
\begin{equation*}
m_0=\delta_0+n_1\delta_1-\delta_2, \;\;\;
m_i=n_{i+1}\delta_{i+1}-\delta_{i+2}
\end{equation*}
for $1\leq i\leq s-2$. These formulae can be deduced from the
knowledge of the sequence of Newton polygons associated with the
singularity at infinity of the curve defined by the equation
$q_{g+1}(x,y)=0$ and results in \cite[IV.3]{cam}.

With respect to the dual graph or the blocks in the HNE of the
germ, one gets \label{launo}
\begin{equation}
\label{z} \frac{m_{j-1}}{e_{j-1}}+k_{j-1}= h_{s_{j-1}}  +
\frac{1}{{h_{s_{j-1}+1}  + _{ \ddots + \frac{1}{{h_{s_j-1}  +
\frac{1}{{k_j }}}}} }}, \end{equation} for $j=1,2,\ldots,g$, where
$s_0=k_0=0$ (see \cite[III.4]{cam}).

We end this section by collecting recent information about a
particularly simple set of curves with only one place at infinity.

\begin{de}
{\rm An  {\it Abhyankar-Moh-Suzuki} ($AMS$ for short) curve $C$ is a
plane curve with only one place at infinity such that it is rational
and nonsingular in its affine part. }
\end{de}

The advantage of this type of curves is that their dual graphs are
easily described \cite{bobadilla} and their associated
$\delta$-sequences in $\mathbb{N}_{> 0}$ are also very easy to
compute: they are those sequences $\{\delta_i\}_{i=0}^g$ of distinct
positive integers such that $\delta_i$ divides $\delta_{i-1}$ for
all $i=1,2,\ldots,g$ and $\delta_g=1$ (this can be deduced from the
proof of \cite[Proposition 2]{paco}).

\section{Valuations of function fields of surfaces}\label{tres}

In this paper, we are concerned with weight functions given by
certain type of valuations close to curves with only one place at
infinity. For this reason, we devote this section to state the main
facts related to valuations we shall need.
See \cite{zar1,zar,abh1,abh2,spi,kiy} for some significant
applications of valuation theory in algebraic geometry. Additional
details to this section, following the same line of this paper, can
be found in Section 3 of \cite{ga-sa}.

From now on, we shall set $p$  a point in $\mathbb{P}^2$, $R :=
\mathcal{O}_{\mathbb{P}^2, p}$ and $K$ the quotient field of $R$.

\begin{de}
{\rm  A {\em valuation} of the field $K$ is a mapping
\[
\nu : K^* (:= K \setminus \{0\}) \rightarrow G,
\]
where $G$ is a totally ordered group, such that it satisfies $\nu (f
+ g) \geq \min \{\nu(f), \nu (g) \}$ and $\nu (fg) =\nu (f) + \nu
(g)$, $f$ and $g$ being elements in $K^*$. }
\end{de}

A valuation as above is said to be centered at $R$ whenever $R
\subseteq R_\nu := \{ f \in K^* | \nu(f) \geq 0\}\cup \{0\}$ and $R
\cap m_\nu \left(:= \{ f \in K^* | \nu(f) > 0\}  \cup \{0\}\right)$
coincides with the maximal ideal $m$ of $R$. We call this type of
valuations {\it plane valuations}. Assume for a while that the field
$k$ is algebraically closed. Plane valuations have a deep
geometrical meaning as the following  result proves (see
\cite{spi}):

\begin{theo}\label{b}
There is a one to one correspondence between the set of
plane valuations (of $K$ centered at $R$) and the set of simple
sequences of quadratic transformations of the scheme {\rm Spec} $R$.
\end{theo}

What Theorem \ref{b} says is that, attached to a plane valuation
$\nu$, there is a unique sequence of point blowing-ups
\begin{equation}
\label{uno}  \cdots X_{N+1} \stackrel{\pi_{N+1}}{\longrightarrow}
X_{N} \longrightarrow \cdots \longrightarrow X_{1}
\stackrel{\pi_{1}}{\longrightarrow} X_{0}=X= {\rm Spec} \;R,
\end{equation}
where $\pi_{i+1}$ is the blowing-up of $X_i$ at the unique closed
point $p_i$ of the exceptional divisor $E_i$ (obtained after the
blowing-up $\pi_i$) satisfying that $\nu$ is centered at the local
ring  $\mathcal{O}_{X_i,p_i}$ $(:= R_i)$. Conversely, each
sequence as in (\ref{uno}) provides a unique plane valuation.

This fact reflects the closeness between plane valuations and germs
of plane curves. It is clear that in a similar way to that explained
in Section 2, we can provide a dual graph and also a HNE (with
respect to a fixed regular system of parameters $\{u,v\}$ of the
ring $R$) for each plane valuation. Notice that, in this case, the
sequence (\ref{uno}) can be either infinite or finite (what does not
happen for germs of curves, although usually it is only showed the
important part, that is the blowing-ups we must do until the germ is
resolved). Attending to the structure of the dual graph of the
sequence (\ref{uno}), Spivakovsky in \cite{spi} classifies the plane
valuations in five types (this refines a previous classification by
Zariski and it can also be refined \cite{ga}). We briefly recall the
essential of this classification, since it will be useful for us.
For further reference in the line of this paper, the reader can see
\cite[3.3]{ga-sa}, where this classification is given in terms of
the HNE of the valuations. Notice that the HNE has the advantage
that it is suitable for positive characteristic  and provides
parametric equations of the valuations. \vspace{2mm}

-- Valuations whose associated sequence (\ref{uno}) is finite are
called of {\bf TYPE A}.

-- A plane valuation whose sequence (\ref{uno}) consists, from one
blowing-up on, only of blowing-ups at free points is named of {\bf
TYPE B}. Sequences (\ref{uno}) for these valuations behave as
those that resolve  germs of plane curves.

-- {\bf TYPE C} valuations are those such that their attached
sequence (\ref{uno}) can contain finitely many blowing-ups as above,
that is, with blocks of free and satellite points alternatively but
it ends with infinitely many satellite blowing-ups, all of them with
center at the strict transform of the same exceptional divisor.

-- When the sequence (\ref{uno}) is  as above, that is, it ends with
infinitely many satellite blowing-ups, but they are not ever
centered at the strict transform of the same exceptional divisor, we
get {\bf TYPE D} valuations.

-- Finally, a valuation whose corresponding sequence (\ref{uno})
alternates indefinitely blowing-ups at (blocks of) free and
satellite points is named to be a {\bf TYPE E} valuation.
\vspace{3mm}

An important fact for the development of this paper is that any
plane valuation  can be regarded as a limit of a sequence
$\{\nu_i\}_{i=1}^\infty$ of type A  valuations. The valuations
$\nu_i$ correspond to the divisors created by the sequence
(\ref{uno}) attached to $\nu$. The proof is based on the fact that
the ring $R_\nu$ is the direct limit of the sequence of rings $R_i$.

When the valuation is centered at a two-dimensional regular local
ring $\mathfrak{R}$ whose residue field is not algebraically
closed, the above procedure works similarly. Indeed, the valuation
ring $\mathfrak{R}_\nu$ of a type A valuation $\nu$ is a local
ring  that dominates $\mathfrak{R}$ and has the same quotient
field as $\mathfrak{R}$. Between $\mathfrak{R}$ and
$\mathfrak{R}_\nu$ there exists a uniquely determined sequence of
mutually dominated local rings
\[
\mathfrak{R} \subset \mathfrak{R}_1 \subset \cdots \subset
\mathfrak{R}_{N+1} =\mathfrak{R}_\nu,
\]
 which provides the blowing-up sequence. As above, the residue field
 of $\mathfrak{R}_\nu$ is a transcendental extension of the one of
 $\mathfrak{R}$ and the difference consists on the fact that the
 residue field of $\mathfrak{R}_i$ ($1 \leq i \leq N$) needs not coincide
 with the one of $\mathfrak{R}$. However, in this paper we only
 consider valuations where the above residue fields always coincide
 and therefore we can handle our valuations as in the algebraically
 closed case.

Returning to the sequence $\{\nu_i\}_{i=1}^\infty$ of type A
valuations converging to a plane valuation $\nu$, when we deal
with a type D or E valuation $\nu$, its value group $G$ is a
subset of the set of real numbers $\mathbb{R}$ and, if we consider
the normalization $\nu'_i := \nu_i / \nu_i (m)$ of the valuations
$\nu_i$, $m$ being the maximal ideal of $R$ and $\nu_i (m)$  the
minimum of the values $\nu_i(g)$ when $g$ runs over $m\setminus
\{0\}$, then $\nu(f) = \lim_{i \rightarrow \infty} \nu'_i (f)$,
for all $f\in K^*$. For type B or C valuations, $G$ is included in
$\mathbb{Z}^2$, $\mathbb{Z}$ denoting the integer numbers, and one
can understand the limit by using the Noether formula. Indeed, let
$\nu$ be  a type A valuation, with attached sequence (\ref{uno})
which ends at $X_{N+1}$ and $f \in R$ an analytically irreducible
element. Set $m_i$ the maximal ideal of the local ring $R_i$ in
the above sequence and assume that $p_0, p_1, \ldots, p_r$, $r
\leq N$, are the common infinitely near points for the sequence
(\ref{uno}) and the resolution of the germ of curve given by $f$.
Then $\nu(m_i)
>0$ can be easily computed from the dual graph or the HNE of $\nu$
(as in the case of germs of curves) and
\begin{equation}
\label{noet}\nu(f) = \sum_{i=0}^r \nu(m_i) e(p_i),
\end{equation}
where $e(p_i)$ is the multiplicity of the germ given by $f$ at the
point $p_i$. When $\nu$ is of type B, we get the same equality,
but here $N = \infty$ and $\nu(m_i) \in \mathbb{Z}^2$ although its
first coordinate vanishes; the limit appears when $r= \infty$
because then we set $\nu(f) =(1,0)$. Finally, Formula (\ref{noet})
also holds whenever $\nu$ is a type C valuation;
\label{tresestrellas} the concept of limit appears in the values
$\nu(m_i)$ in the following way: set $p_{i_0}$ that blowing-up
center in (\ref{uno}) having infinitely many proximate points
$p_i$, $i>i_0$, then $\nu(m_i)=(0,1)$ when $i>i_0$,
$\nu(m_{i_0})=(1,0)$ (the ``sum" of infinitely many $(0,1)$) and
the remaining values $\nu(m_i)$ can be computed from the above
ones as in the case of type A valuations.

\section{Weight functions given by plane valuations at infinity}
\label{cuatro}

\subsection{Weight functions and plane valuations at infinity}

To give the definition of weight function, first we recall some
concepts related to semigroups. Assume that $\alpha, \beta,
\gamma$ are arbitrary elements in a commutative semigroup  with
zero $\Gamma$. If $\leq$ is an ordering on  $\Gamma$, $\leq$ is
said to be {\it admissible} if $0 \leq \alpha$ and, moreover,
$\alpha \leq \beta$ implies $\alpha + \gamma \leq \beta + \gamma$.
On the other hand,  $\Gamma$ is named {\it cancellative} whenever
from the equality $\alpha + \beta = \alpha + \gamma$ one can
deduce $\beta = \gamma$. Finally for stating the mentioned
definition, stand $\Gamma$ for a cancellative well-ordered
commutative with zero and with admissible ordering semigroup and
$\Gamma \cup \{-\infty\}$ for the above semigroup together with a
new minimal element denoted by $-\infty$, which satisfies $\alpha
+ (- \infty) = - \infty$ for all $\alpha \in \Gamma \cup
\{-\infty\}$.

\begin{de}
{\rm  A {\em weight function} from a $k$-algebra $A$ onto a
semigroup $\Gamma \cup \{-\infty\}$ as above is a mapping $w: A
\longrightarrow \Gamma \cup \{-\infty\}$ such that, for $p, q \in
A$, the following statements must be satisfied:
\begin{enumerate}
\item $w(p) = - \infty$ if and only if  $p=0$; \item $w(ap) = w
(p)$ for all nonzero element $ a \in k$; \item $w(p+q) \leq \max \{
w(p),w(q) \}$; \item If $ w(p) = w(q)$, then there exists a nonzero
element $a \in k^*$ such that $w(p-aq)< w(q)$; \item $w(pq)= w(p) +
w(q)$.
\end{enumerate}
}
\end{de}

When the last condition is not imposed, we get the definition of
{\it order function}. It is clear that if $w$ is a weight
function, then the triple $(A,w,\Gamma)$ is an order domain over
$k$ (see, for instance, \cite{g-p} for the definition of order
domain).

The next result, proved in \cite[Proposition 2.2]{ga-sa}, shows  how
to get weight functions from valuations.

\begin{prop}
\label{a} Let $\mathfrak{K}$ be the quotient field of a Noetherian
regular local domain $\mathfrak{R}$ with maximal $\mathfrak{m}$.
Let $\nu: \mathfrak{K}^* \rightarrow \mathfrak{G}$ be a valuation
of $\mathfrak{K}$ which is centered at $\mathfrak{R}$. Assume that
the canonical embedding of the field $\mathfrak{k} =
\mathfrak{R}/\mathfrak{m}$ into the field $ \mathfrak{R}_\nu
/\mathfrak{m}_\nu$ is an isomorphism.

Set $w: \mathfrak{K}^* \rightarrow \mathfrak{G}$ the mapping given
by $w(f) = - \nu (f)$, $f \in \mathfrak{K}^*$. If $\mathfrak{A}
\subseteq \mathfrak{K}^*$ is a $\mathfrak{k}$-algebra such that
$w(\mathfrak{A})$ is a cancellative, commutative, free of torsion,
well-ordered semigroup with zero, $\Gamma$, where the associated
ordering is admissible, then $w: \mathfrak{A} \longrightarrow
w(\mathfrak{A}) \cup \{- \infty\}$, $w(0) = - \infty$, is a weight
function.
\end{prop}

Recall that $p \in \mathbb{P}^2$, $R=\mathcal{O}_{\mathbb{P}^2,p}$
and $K$ is the quotient field of $R$. The isomorphism $R/m \cong
R_\nu /m_\nu$ happens for any plane valuation $\nu$ except for
those of type A.

Now, we are going to introduce another fundamental concept for us:
{\it plane valuation at infinity}. To do it, we start by stating
the concept of general element of a  type A plane valuation. As we
have said, these valuations are the unique whose corresponding
sequence of point blowing-ups is finite. In fact, they are defined
by the last created exceptional divisor and, by this reason, they
are also named divisorial valuations. Concretely, with the
notation in Section 3, if $\pi_{N+1}$ is the last blowing-up in
the sequence (\ref{uno}) given by a divisorial valuation $\nu$,
then $\nu$ is the $m_N$-adic valuation, $m_N$ being the maximal
ideal of the ring $R_N$.

\begin{de}
{\rm Let $\nu$ be a divisorial valuation. An element $f$ in the
maximal ideal of $R$ is called to be a {\it general element of
$\nu$} if the germ of curve given by $f$ is analytically
irreducible, its strict transform in $X_{N+1}$ is smooth and meets
$E_{N+1}$ transversely at a non-singular point of the exceptional
divisor of the sequence (\ref{uno}) attached to $\nu$. }
\end{de}

General elements are useful to compute plane valuations. Indeed,
if $f \in R$, then
$$
\nu (f) = \mbox{ min $\{(f,g) | g $ is a general element of $\nu
\},$ }
$$
where $(f,g)$ stands for the intersection multiplicity of the
germs of curve given by $f$ and $g$.

\begin{de}
{\rm A {\it plane divisorial valuation at infinity} is a plane
divisorial valuation of $K$ centered at $R$ that admits, as a
general element, an element in $R$ providing the germ at $p$ of some
curve with only one place at infinity ($p$ being its point at
infinity). }
\end{de}

\begin{de}
\label{preaprva} {\rm A  plane {\it valuation} $\nu$ of $K$
centered at $R$ is said to be {\it  at infinity} whenever it is a
limit of plane divisorial valuations at infinity. More explicitly,
set $\{\nu_i\}_{i=1}^\infty$ the set of plane divisorial
valuations, corresponding to divisors $E_i$, that appear in the
sequence (\ref{uno}) given by $\nu$. $\nu$ will be at infinity if,
for any index $i_0$, there is some $i > i_0$ such that $\nu_i$ is
a plane divisorial valuation at infinity. }
\end{de}

Afterwards, we shall see that, as it happens for type A
valuations, type B ones are not suitable for our purposes. So we
exclude them of the next definition.

\begin{de}
\label{aprva} {\rm Let $\nu$ be a plane valuation at infinity that
is neither of type A nor of type B. A sequence of polynomials $P=\{
q_i(x,y) \}_{i \geq 0}$ in $k[x,y]$ is a {\it family of
approximates} for $\nu$ whenever each plane curve $C$ with only one
place at infinity providing a general element of some of the plane
divisorial valuations at infinity converging to $\nu$ admits some
subset of $P$ as a family of approximates and $P$ is minimal with
this property.}
\end{de}

\subsection{$\delta$-sequences }
\label{42}

Now, we introduce the most important concept of this paper: the one
of $\delta$-sequence. First of all, notice that if $S$ is an ordered
commutative semigroup in $\mathbb{R}$ (naturally ordered) or
$\mathbb{Z}^2$ (with the lexicographical ordering, that is $\alpha
> \beta$ if and only if the left-most nonzero entry of $\alpha -
\beta$ is positive) and $\gamma_0 > \gamma_1>0$ are in $S$ such
that $j \gamma_1 > \gamma_0$ for some integer number $j >0$, then
we can set $\gamma_0 = m_1 \gamma_1 + \gamma_2$ with $m_1 \in
\mathbb{N}_{>0}$, $\gamma_2$ in $\mathbb{R}$ or $\mathbb{Z}^2$ and
$0 \leq \gamma_2 < \gamma_1$ (note that $m_1$ and $\gamma_2$ are
uniquely determined by these conditions). Repeating this
procedure, one gets $\gamma_1 = m_2 \gamma_2 + \gamma_3$ whenever
$j \gamma_2
> \gamma_1$ for some $j >0$. If we iterate, there are three
possibilities: {\it Case} 1) The algorithm stops because we get
$\gamma_i =0$ for some index $i$; in this case we can speak about
the {\it greatest common divisor} $\gamma_{i-1}$ of $\gamma_0$ and
$\gamma_1$. {\it Case} 2) The algorithm does not stop.  {\it
Case}~3) In certain step $i$, there is no $j \in \mathbb{N}_{>0}$
such that $j \gamma_i
> \gamma_{i-1}$. The concept of $\delta$-sequence is motivated by
these three possibilities, that is, we shall consider sequences of
positive elements $\Delta= \{ \delta_0,\delta_1,\ldots,\delta_i,
\ldots \}$ in $\mathbb{R}$ or in $\mathbb{Z}^2$ and,  reproducing
the computations in pages \pageref{z0} and \pageref{z} with the help
of the above procedure (that is the Euclidian algorithm adapted to
our data), we shall arrive to different situations which allow us to
enlarge in a natural way the concept of $\delta$-sequence in
$\mathbb{N}_{>0}$.

For a start, a {\it normalized $\delta$-sequence in
$\mathbb{N}_{>0}$} is an ordered finite set of rational numbers
$\overline{\Delta} = \{
\overline{\delta}_0,\overline{\delta}_1,\ldots,\overline{\delta}_g
\}$ such that there is a $\delta$-sequence in $\mathbb{N}_{>0}$,
$\Delta = \{ \delta_0,\delta_1,\ldots,\delta_g \}$, satisfying
$\overline{\delta}_i = \delta_i /\delta_1$ for $0 \leq i \leq g$.
Notice that from $\overline{\Delta}$, by writing
$\overline{\delta}_i = r_i /s_i$ as a quotient of relatively prime
elements, one can recover $\Delta$ since $\delta_i =
\overline{\delta}_i \mbox{lcm}(s_i)_{0 \leq i \leq g}$. From now on,
we set $C_\Delta = C_{\overline{\Delta}}$  the curve with only one
place at infinity given by $\Delta$ that provides Proposition
\ref{siete}. The dual graph of the resolution of the singularity of
$C_\Delta$ at the point at infinity will be named {\it the dual
graph given by $\Delta$ or $\overline{\Delta}$.}

\begin{de}
\label{buena} {\rm A {\it $\delta$-sequence} in $\mathbb{Z}^2$
(respectively, $\mathbb{Q}$) (respectively, $\mathbb{R}$) is a
sequence $\Delta = \{ \delta_0,\delta_1,\ldots,\delta_i, \ldots \}$
of elements in $\mathbb{Z}^2$ (respectively, $\mathbb{Q}$)
(respectively, $\mathbb{R}$) such that it generates a well-ordered
sub-semigroup of $\mathbb{Z}^2$ (respectively, $\mathbb{Q}$)
(respectively, $\mathbb{R}$) and
\begin{description}
\item[($\mathbb{Z}^2$)] $\Delta = \{
\delta_0,\delta_1,\ldots,\delta_g \}
 \subset \mathbb{Z}^2$ is finite, $g \geq 2$ (respectively, $\geq 3$) and there exists a
 $\delta$-sequence in $\mathbb{N}_{>0}$,
 $\Delta^* = \{ \delta_0^*,\delta_1^*,\ldots,\delta_g^* \}$, such
 that $\delta_0^*-\delta_1^*$ does not divide (respectively, divides)
 $\delta^*_0$ and
\[
\delta_i = \frac{\delta_i^*}{A a_t +B} (A,B) \;\;\; (0\leq i \leq
g-1)\;\; \mbox{ and}
\]
\[
\delta_g = \frac{\delta_g^*+A' a_t +B' }{A a_t +B} (A,B) -(A',B'),
\]
where $\langle a_1;a_2, \ldots,a_t\rangle$, $a_t\geq 2$,  is the
continued
 fraction expansion of the quotient $m_{g-1}/ e_{g-1}$
 (respectively, $m_{g-2}/ e_{g-2}$) given by
 $\Delta^*$ (see (\ref{z})) and, considering the finite recurrence
 relation $\underline{y}_i = a_{t-i}\underline{y}_{i-1} +
 \underline{y}_{i-2}$, $\underline{y}_{-1}=(0,1)$,
 $\underline{y}_{0}=(1,0)$, then $(A,B) := \underline{y}_{t-2}$ and $(A',B') : =
 \underline{y}_{t-3}$. We complete this definition by adding that $\Delta = \{
 \delta_0,\delta_1\}$ (respectively, $\Delta = \{
 \delta_0,\delta_1,\delta_2\}$) is a $\delta$-sequence in $\mathbb{Z}^2$
 whenever $\delta_0 = \underline{y}_{t-1}$ and $\delta_0 - \delta_1 =
 \underline{y}_{t-2}$ (respectively, $\delta_0=j\underline{y}_{t-2}$,
 $\delta_0-\delta_1=\underline{y}_{t-2}$ and $\delta_0+n_1\delta_1-\delta_2=\underline{y}_{t-1}$) for the above recurrence attached to a
 $\delta$-sequence in $\mathbb{N}_{>0}$,
 $\Delta^* = \{ \delta_0^*,\delta_1^*\}$ (respectively, $\Delta^* = \{
 \delta_0^*,\delta_1^*,\delta_2^*\}$, such that $j:=\delta_0^*/(\delta_0^*-\delta_1^*)\in \mathbb{N}_{\geq 0}$ and $n_1:=\delta_0^*/\gcd(\delta_0^*,
 \delta_1^*)$).

\item[($\mathbb{Q}$)] (Respectively, $\Delta =
\{\delta_0,\delta_1,\ldots,\delta_i, \ldots\} \subset \mathbb{Q}$
is infinite and any
        ordered subset $\Delta_j = \{\delta_0,\delta_1,\ldots,\delta_j\}$
        is a normalized $\delta$-sequence in $\mathbb{N}_{>0}$).

\item[($\mathbb{R}$)] (Respectively, $\Delta = \{
\delta_0,\delta_1,\ldots,\delta_g \} \subset \mathbb{R}$ is
        finite, $g \geq 2$, $\delta_i$ is a positive rational number for $0 \leq i \leq
        g-1$, $\delta_g $  is non-rational,
        and there exists a sequence $$ \left \{ \overline{\Delta}_j =
        \{ \delta_0^j,\delta_1^j,\ldots,\delta_g^j \} \right\}_{j
        \geq1}$$
        of normalized $\delta$-sequences in $\mathbb{N}_{>0}$ such
        that $\delta_i^j = \delta_i$ for $0 \leq i \leq
        g-1$ and any $j$ and  $\delta_g = \lim_{j \rightarrow \infty}
        \delta_g^j$. We complete this definition by adding that
        $\Delta = \{\tau,1\}$, $\tau>1$ being a non-rational number,
        is also a $\delta$-sequence in $\mathbb{R}$).
\end{description}
}
\end{de}

For simplicity's sake, we shall do our theoretical development only
for $\delta$-sequences (in $\mathbb{Z}^2$, $\mathbb{Q}$ or
$\mathbb{R}$) verifying that there is no positive integer $j$ such
that $\delta_0=j(\delta_0-\delta_1)$. However all the results in
this paper are true for any $\delta$-sequence because the reasonings
in the non-considered case are similar taking into account that
$m_{g-2}$ and $e_{g-2}$ must be used instead $m_{g-1}$ and
$e_{g-1}$.

Let us see that $\delta$-sequences $\Delta$ in $\mathbb{Z}^2$ are
intimately related to type C plane valuations at infinity. Consider
the curve $C:= C_{\Delta^*}$ and the set of associated polynomials
$\{q_i(x,y)\}_{i=0}^{g+1}$ of  Proposition \ref{siete}. Then,
$\delta^*_i=-\nu_{C,p}(q_i(x,y))$ for all $i=0,1,\ldots,g$. We have
defined the values $\delta_i$ in the same form, but associated to a
certain valuation at infinity of type C obtained from $C$. Indeed,
let $ \mathcal{D} := \{p_0=p,p_1,p_2,\ldots\}$ be the sequence of
centers of the blowing-ups associated with the curve $C$ given in
(\ref{infiniteseq}) and set $i_0$ the maximum  among the positive
integers $j$ such that $p_j$ admits more than a point proximate to
it; our definition uses the relationship between $ -
\nu_{C,p}(q_i(x,y)) = \delta_i^*$ and $ - \nu(q_i(x,y))$, which we
name $\delta_i$, where $\nu$ is the valuation of type C defined by
the infinite sequence of quadratic transformations of the scheme
${\rm Spec }\;{\mathcal O}_{{\mathbb P}^2,p}$ centered at the closed
points in the set ${\mathcal C}:=\{r_j\}_{j\in {\mathbb N}}$, where
$r_j:=p_j$ whenever $j\leq i_0+1$ and, for each $j>i_0+1$, $r_j$ is
the unique point of the blowing-up centered at $r_{j-1}$ which is
proximate to $r_{i_0}=p_{i_0}$.

The concrete relation that we give in the definition can be deduced
as follows. The integer $a_t$  is the number of points in $\mathcal
D$ which are proximate to $r_{i_0}$; then $e_{i_0+1}(C)=1$ and
$e_{i_0}(C)=a_t$, where $e_{j}(C)$ denotes the multiplicity at $r_j$
of the strict transform of the germ $(C,p)$ at $r_j$. The remaining
multiplicities can be obtained using recurrent relations, the
so-called {\it proximity equalities}, that is, $e_{j}(C)=\sum_k
e_{k} (C)$ for all $j\geq 0$, where $k$ runs over the set of indexes
such that $r_k$ is proximate to $r_j$. The values $\nu(m_j)$ satisfy
the same relations, but with different initial values:
$\nu(m_j)=(0,1)$ for all $j>i_0$, $\nu(m_{i_0})=(1,0)$ and
$\nu(m_j)=\sum_k \nu(m_k)$ for $j<i_0$, with $k$ running over the
same set as before (recall the last paragraph of Section
\ref{tres}). From these facts it is easy to deduce that, if $i_1$
denotes the maximum index such that $r_{i_1}$ is a free point but
$r_{i_1+1}$ is not so, there exist natural numbers $A',B',A,B$ such
that $e_{i_1}(C)=A'a_t+B'$, $e_{i_1-1}(C)=Aa_t+B$,
$\nu(m_{i_1})=(A',B')$ and $\nu(m_{i_1-1})=(A,B)$. The first two
values are the integers ${y}_{t-3}$ and ${y}_{t-2}$ (${y}_{t-2}$ and
${y}_{t-1}$ in the case $g=1$) obtained from the recurrence relation
given in Definition \ref{buena} by taking the initial values
${y}_{-1}=1$ and ${y}_{0}=a_t$ (see \cite[III.4]{cam}); so, the last
two values can be obtained as in the definition.  Taking coordinates
$(u,v)$ around the point at infinity $p$ (as in Definition
\ref{aprox}) one has that $q_0(x,y)=v^{-1}$ and, for $i \geq1$,
\begin{equation}\label{infinito}
q_i(x,y) = v^{-\delta_0^*/d_i^*} \overline{q}_i(u,v),
\end{equation}
$\overline{q}_i(u,v)$ being the local expression around $p$ of the
curve given by $q_i(x,y)$ and $d_i^* = \gcd(\delta_0^*,
\delta_1^*, \ldots, \delta_{i-1}^*)$.  Therefore,
\begin{equation}\label{infinito2}
\delta_0=\nu(v) \mbox{ and }
\delta_i=-\nu(q_i(x,y))=\frac{\delta_0^*}{d_i^*}\delta_0-\nu(\overline{q}_i(u,v))
\mbox{ for } i\geq 1.
\end{equation}
 From the above information and Formula (\ref{noet}) applied to the germs
 defined by $\overline{q}_i(u,v)$ it can be deduced the existence of natural
 numbers $\{b_i\}_{i=0}^g$ such that
 $\nu_{C,p}(\overline{q}_g(u,v))=b_g(Aa_t+B)+A'a_t+B'$,
 $\nu(\overline{q}_g(u,v))=b_g(A,B)+(A',B')$ and, for each $i\leq g-1$,
 $\nu_{C,p}(\overline{q}_i(u,v))=b_i(Aa_t+B)$  and $\nu(\overline{q}_i(u,v))=b_i(A,B)$.
 As a consequence, the relations between $\delta_i$ and $\delta_i^*$ given in the
 Definition \ref{buena} hold by taking  (\ref{infinito2}) into account.
\vspace{2mm}

We note that if $\Delta^* = \{
\delta_0^*,\delta_1^*,\ldots,\delta_g^* \}$ and $\Delta' = \{
\delta'_0,\delta'_1,\ldots,\delta'_g \}$ are $\delta$-sequences in
$\mathbb{N}_{>0}$ such that $\delta^*_i /
\delta^*_1=\delta'_i/\delta'_1$ ($0\leq i<g$) and the continued
fraction expansions $\langle a_1;a_2, \ldots,a_t\rangle$ of the
respective quotients $m_{g-1}/e_{g-1}$ (see (\ref{z})) have the
same length and only differ in the last value $a_t$,  then the
$\delta$-sequences in $\mathbb{Z}^2$ defined by $\Delta^*$ and
$\Delta'$ will be equal. The reason is that the above numbers
$b_i$ depend only on the proximity relations among the points
$r_j$ with $j\leq i_0+1$ (notice that the index $i_0$ is clearly
the same for $\Delta^*$ and $\Delta'$) and these coincide for both
$\delta$-sequences in $\mathbb{N}_{>0}$. \vspace{2mm}

To guarantee that our valuation is a plane valuation at infinity,
we must prove that if $\Delta^* = \{\delta_0^*, \delta_1^*,
\ldots, \delta_g^* \} $ is a $\delta$-sequence in
$\mathbb{N}_{>0}$, with attached continued fraction $\langle
a_1;a_2, \ldots,a_t \rangle$ (see Definition \ref{buena}), then
there exists another $\delta$-sequence in $\mathbb{N}_{>0}$,
$\Delta' = \{ \delta'_0, \delta'_1, \ldots, \delta'_g \}$ such
that $ \delta_i^* / \delta_1^* = \delta'_i / \delta'_1$ ($0 \leq i
< g$) and the associated continued fraction $\langle a'_1;a'_2,
\ldots,a'_t \rangle$ satisfies $a_i=a'_i$, $1 \leq i < t$ and $a_t
< a'_t$ (note that the first condition implies that $C_{\Delta^*}$
and $C_{\Delta'}$ can be chosen with the same set of approximates;
this fact and the  second condition guarantee that the divisorial
valuations associated to the resolution of the singularity of
$C_{\Delta^*}$ at infinity are also associated to the one of
$C_{\Delta'}$). Let us see how to get $\Delta'$.

Set $\{
\check{\delta}_0,\check{\delta}_1,\ldots,\check{\delta}_{g-1} \}$
the $\delta$-sequence in $ \mathbb{N}_{>0}$ corresponding to the
normalized sequence $\{ \delta_i^*/\delta_1^* \}_{i=0}^{g-1}$ and
define $\delta'_i = z \check{\delta}_i$, $0 \leq i \leq g-1$,  for
some $z \in \mathbb{N}_{>0}$ to be defined later. Consider the
sequence of convergents \cite[Chapter 7]{ni-zu} $\{h_n,k_n\}$ for
the continued fraction $\langle a_1; a_2, \ldots, a_t \rangle$. For
$a \in \mathbb{N}_{>0}$, one gets
\[
 \langle a_1; a_2,
\ldots, a_{t-1},a \rangle = \frac{a h_{t-1} + h_{t-2}}{a k_{t-1} +
k_{t-2}}.
\]
Then
\[
 \frac{a h_{t-1} + h_{t-2}}{a k_{t-1} +
k_{t-2}} = \frac{m'_{g-1}}{e'_{g-1}} = \frac{m'_{g-1}}{z},
\]
$m'_{g-1}$ and $e'_{g-1}$ being the values defined in page
\pageref{z0}  for our tentative  $\Delta'$ and
\[
\delta'_g = n'_{g-1} z \check{\delta}_{g-1} - m'_{g-1}=
\]
\[
= n'_{g-1} (a k_{t-1} + k_{t-2}) \check{\delta}_{g-1} -(a h_{t-1}
+ h_{t-2}) =
\]
\[
= k_{t-1} (n'_{g-1} \check{\delta}_{g-1} -
\frac{h_{t-1}}{k_{t-1}}) a + k_{t-2} (n'_{g-1}
\check{\delta}_{g-1} - \frac{h_{t-2}}{k_{t-2}}).
\]
So, we only need to pick $z$ large enough to that $a'_t =a > a_t$
and $\delta_g \in \langle \check{\delta}_0, \check{\delta}_1,
\ldots, \check{\delta}_{g-1} \rangle$, which is possible since gcd$(\check{\delta}_0, \check{\delta}_1, \ldots,
\check{\delta}_{g-1}) = 1$ and the semigroup that these elements
generate $\langle \check{\delta}_0, \check{\delta}_1, \ldots,
\check{\delta}_{g-1}
\rangle$ has a conductor.\\

$\delta$-sequences in $\mathbb{Q}$ (respectively, in $\mathbb{R}$)
are related with valuations at infinity of type E (respectively, D).
To see it, it suffices to recall the definition and, in the first
case, to consider the valuation given by the sequence of infinitely
near points associated with the curves given by the polynomials
$q_i$ mentioned in page \pageref {z0}  for $i$ large enough. In the
second case (assuming $g\geq 2$), one must consider one of the
normalized $\delta$-sequences in $\mathbb{N}_{>0}$
$\overline{\Delta}_j$ of the definition and its corresponding
$\delta$-sequence in $\mathbb{N}_{>0}$, say
$\Delta_j=\{\delta'_0,\delta'_1,\ldots,\delta'_g\}$; the related
valuation of type D is determined by the sequence of infinitely near
points associated with the resolution of the singularity at infinity
of the curve defined by the approximate $q_g(x,y)$ of
$C_{\Delta_j}$, and the infinitely many satellite points whose
corresponding block $B_g$  is determined by the continued fraction
expansion  of the non-rational number
$\frac{n_{g-1}\delta_{g-1}-\delta_g}{e_{g-1}/\delta'_1}$ (see pages
\pageref{z0} and \pageref{z}). Observe that the numbers $n_{g-1}$
and $e_{g-1}/\delta'_1$ do not depend on the chosen
$\delta$-sequence $\Delta_j$; in fact, the above mentioned
non-rational number is the limit of the sequence of quotients
$\frac{m_{g-1}}{e_{g-1}}$ associated with the $\delta$-sequences
$\Delta_j$. With respect to the case $g=2$ notice that, for
whichever non-rational number $\tau>1$, from the convergents of the
continued fraction given by $\tau/(\tau-1)$ we can derive normalized
$\delta$-sequences $\Delta^j = \{\delta_0^j, 1\}$ approaching the
$\delta$-sequence in $\mathbb{R}$, $\{\tau,1\}$. The described
valuations associated to $\delta$-sequences in $\mathbb{Q}$ are
clearly valuations at infinity and this is also for the real case by
similar reasonings to those given for $\delta$-sequences in
$\mathbb{Z}^2$.

As a consequence of the last paragraphs, we have proved the
following

\begin{prop}
\label{R} Let $\Delta = \{ \delta_i\}_{i=0}^r, r \leq \infty$, be a
$\delta$-sequence in $\mathbb{Z}^2$ (respectively, in $\mathbb{Q}$)
(respectively, in $\mathbb{R}$). Then,  there exists a plane
valuation at infinity $\nu_\Delta$ of type C (respectively, E)
(respectively, D) and a family $\{q_i(x,y)\}_{i=0}^r$ of
approximates for $\nu_{\Delta}$ such that $- \nu_\Delta (q_i(x,y)) =
\delta_i$ for all index $i$.
\end{prop}

Notice that we have chosen a concrete valuation $\nu_\Delta$, but
this election needs not be unique since we could take another
suitable families of approximates. Moreover, even when we have fixed
certain type of approximates, we have infinitely many possibilities
for our $q_i$, $1<i<r$, according the parameters $t_i$ we set in
(\ref{misq}).

\noindent {\it Remark}. In this paper, generically, we shall name
$\delta$-sequence to any of the $\delta$-sequences (in
$\mathbb{Z}^2$, $\mathbb{Q}$ or $\mathbb{R}$) above defined. Due to
their definition, we can apply the Euclidian algorithm as we
described at the beginning of this subsection to  the values $m_j$,
$e_j$, $0 \leq j < g$, ($g=\infty$ in the case in $\mathbb{Q}$)
defined as in page \pageref{z0} and computed from the
$\delta$-sequences (newly with the Euclidian algorithm). Case 3)
happens for $\delta$-sequences in $\mathbb{Z}^2$, since we get
$\gamma_{i-1}=(1,0)$ and $\gamma_{i}=(0,1)$, Case 2) holds for
$\delta$-sequences in $\mathbb{R}$ and, for $\delta$-sequences in
$\mathbb{Q}$, we have an indefinite iteration of Case 1) .  Along
this paper, we shall use freely the notation $m_j$, $e_j$, $d_j$, $0
\leq j \leq g-1$, adapted to the corresponding $\delta$-sequence.

Table \ref{resumen} summarizes briefly the relation among the
different types of $\delta$-sequences and their corresponding
families of approximates and valuations.

\begin{table}[h,t,b]
\caption[]{Summarizing Table} \label{resumen}
\begin{tabular}{|c|c|c|c|}
  \hline
Type of & Case in &
Family of approximates & Type of \\
$\delta$-sequence & Euclidian &  (given by
$\Delta$ using  & valuation \\
$\Delta$ & algorithm &  the Euclidian algorithm) &(at infinity)\\
  \hline
   $\mathbb{Z}^{2}$& 3 &  Finite &
   C\\  \hline
  $\mathbb{R}$& 2 & Finite &
   D \\  \hline
$\mathbb{Q}$& 1 & Infinite
  & E \\
  \hline
\end{tabular}
\end{table}

\noindent {\it Remark}. Notice that we have no definition of
$\delta$-sequence related to type B valuations. This concept could
be defined following the same line of this paper; however it would
not be useful for us, since we would get $\delta$-sequences $\Delta
= \{ \delta_0,\delta_1,\ldots,\delta_{g+1} \}$ where $\delta_i \in
\{0\} \oplus \mathbb{N}_{>0}$ ($0 \leq i \leq g$) and $\delta_{g+1}
=(-1,a)$, $a \in \mathbb{N}_{>0}$, and then $S_\Delta$ would not be
well-ordered (for the lexicographical ordering).

Next, we state the main result of the paper.

\begin{theo}
\label{main} Let $\Delta= \{ \delta_i\}_{i=0}^r$, $r \leq \infty$,
be a $\delta$-sequence. Set $k[x,y]$ the polynomial ring in two
indeterminates over an arbitrary field $k$. Then,
\begin{itemize}
    \item[(a)] There exists a weight function $w_\Delta: k[x,y]
    \longrightarrow S_\Delta \cup \{- \infty\}$.
    \item[(b)] The map $-w_\Delta: k(x,y)
    \rightarrow G(S_\Delta) \cup \{\infty\}$, $G(S_\Delta)$ being the group generated by $S_\Delta$, is a
    plane valuation at infinity.
    \item[(c)] If $\{q^{,}_i\}_{i=0}^{r}$ denotes a family of approximates for
    the valuation $-w_\Delta$ then, for any $\alpha \in S_\Delta$,
    the vector spaces
    \[
O_\alpha :=  \left \{p \in k[x,y] \mid w_\Delta (p) \leq \alpha
\right \}
    \]
    are spanned by the set of polynomials $\prod_{i=0}^m
    q_{i}^{,\gamma_i}$ such that $0 \leq m < r+1$, $\beta:=\sum_{i=0}^m
    \gamma_i \delta_i$ runs over the unique expression of
    the values $\beta \in S_\Delta$ satisfying $\beta \leq \alpha$,
    $ \gamma_0 \geq 0$,  $0 \leq
    \gamma_i < n_i$, whenever $1 \leq i < m$ and $\gamma_m \geq 0$ if $m=r$ and
    otherwise $0 \leq \gamma_m < n_m$.
\end{itemize}
\end{theo}
\begin{proof}
Let $\nu = \nu_\Delta$ be the plane valuation defined in Proposition
\ref{R}.  Set $w_\Delta := - \nu$. Since $\nu$ is of type C, D or E,
to prove (a) and (b), it suffices to show that $ S_\Delta \cup \{
\infty \}$ is the image of $k[x,y]$ by $w_\Delta$. To do it, recall
that the polynomials $q_i$ of Proposition \ref{R} described before
it satisfy the equalities given in (\ref{infinito}). Consider the
same notation and pick $f(x,y) \in k[x,y]$; then $f(x,y) =
v^{-\deg(f)} \overline{f}(u,v)$ and, since the set $\{ v \} \cup
\{\overline{q}_i(u,v)\}_{i=1}^r$ is a minimal generating sequence
for the valuation $\nu$ (see \cite{spi}), factoring
$\overline{f}(u,v)$ into product of analytically irreducible
elements, one has that there exists a polynomial of the form $m(u,v)
= u^{s_0} v^{s_1} \prod_{i=2}^j \overline{q}_i^{s_i} (u,v)$, $s_i
\geq 0$, $0\leq i \leq j$, $j < r+1$ such that
$\nu(\overline{f}(u,v)) = \nu(m(u,v))$ and $s_0 + s_1 + \sum_{i=2}^j
s_i \deg ( \overline{q}_i) \leq \deg (f)$. Therefore $v^{-\deg(f)}
\overline{f}(u,v)$ has the same valuation as $x^{\deg(f) -(s_0 + s_1
+ \sum_{i=2}^j s_i \deg ( \overline{q}_i))} y^{s_0} \prod_{i=2}^j
{q}_i^{s_i}(x,y)$ and so $-\nu(f) \in S_\Delta$.

Finally, (c) is clear from the forthcoming Proposition \ref{teles}
and the fact that $w_\Delta$ is a weight function and the images by
$w_\Delta$ of the considered products $\prod_{i=0}^m
q_i^{,\gamma_i}$ give exactly once each value $\beta \in
    S_\Delta$ such that $\beta \leq \alpha$.
\end{proof}

A {\it $\delta$-sequence} $\Delta$ will also be said of {\it type}
C, D or E whenever the corresponding valuation $-w_\Delta$ is of
that type. We end this subsection by proving that if one considers
suitable value semigroups, then any weight function of the
polynomial ring $k[x,y]$ comes from a valuation at infinity.

\begin{prop}
\label{peso} Let $w: k[x,y] \rightarrow S$ be a weight function on
a semigroup $S$ such that $S=S_\Delta$ for some $\delta$-sequence
$\Delta$. Then, there exists a plane valuation at infinity $\nu:
k(x,y) \rightarrow G$ such that $-\nu$ and $w$ coincide on the
ring $k[x,y]$.
\end{prop}
\begin{proof}
Consider coordinates in $\mathbb{P}^2$ and local coordinates $(x,y)$
and $(u,v)$ as we gave before Definition \ref{aprox}. It holds that
the natural extension $-w= \nu : k(x,y) \rightarrow G(S)$, where
$G(S)$ is the group generated by $S$, is a plane valuation of type
C, D or E centered at  the local ring $k[u,v]_{(u,v)}$. Pick
polynomials $q_i(x,y)$ such that  $w(q_i)= \delta_i$, $0 \leq i <
r+1$. When $\nu$ is of type E, the curves in $\mathbb{P}^2$ given by
$q_i$ guarantee that $\nu$ is as desired. Finally, in cases C and D,
if we consider $\delta$-sequences in $\mathbb{N}_{>0}$, $\Delta_j =
\{ \delta_0,\delta_1,\ldots,\delta_g \}$, approaching $\Delta$ (as
we have described in Definition \ref{buena} for type D valuations
and after that definition  for type C ones), then the curves in
$\mathbb{P}^2$ given by the polynomials $q_g^{n_g} -
\prod_{i=0}^{g-1} q_i^{a_{gi}}$, where $a_{gi}$ are the unique
coefficients of the expression $n_g \delta_g = \sum_{i=0}^{g-1}
a_{gi} \delta_i$, $0 \leq a_{gi} < n_i$ $(1 \leq i \leq g-1)$, prove
that $\nu$ is a plane valuation at infinity.
\end{proof}

\subsection{Construction of weight functions attached to valuations at infinity}
\label{construction} To end this section, we provide algorithms to
get $\delta$-sequences of  every described type. Recall that the
ordering of the semigroup  $S_\Delta$ is given by the
lexicographical one in $\mathbb{Z}^2$ for type C weight functions
and by the natural ordering in $\mathbb{R}$ (respectively,
$\mathbb{Q}$) for type D (respectively, type E) weight functions.
To get a basis of the vector space $O_{\alpha}$, $\alpha \in
S_\Delta$, we only need to compute approximates $\{q_i\}_{i=1}^r$
for $-w_\Delta$ as we have described to prove Proposition \ref{R}
(see (\ref{misq})) and fix a unique polynomial $q_\beta :=
\prod_{i=1}^m q_i^{\gamma_i}$ with $\beta= \sum_{i=1}^m \gamma_i
\delta_i$ for each $\beta \in S_\Delta$ such that $\beta \leq
\alpha$. Then, the set $\{q_\beta\}_{\beta \leq \alpha}$ will be a
basis as desired.

\subsubsection{$\delta$-sequences in
$\mathbb{N}_{>0}$.}\label{ene}

The concept of $\delta$-sequence in $\mathbb{N}_{>0}$ is an
important tool in this paper because it supports our general
definition of $\delta$-sequence. In \cite{pa21}, it can be found a
complete list of $\delta$-sequences in $\mathbb{N}_{>0}$ for curves
over $\mathbb{C}$ with only one place at infinity and genus $\leq
30$ (notice that, in that list, one must interchange $\delta_0$ and
$\delta_1$ in order to be coherent with our notation). In the same
paper, it is announced the existence of an algorithm for obtaining
any $\delta$-sequence in $\mathbb{N}_{>0}$. Next, we present an
algorithm whose input is a $\delta$-sequence $\Delta
=\{\delta_0,\delta_1,\ldots,\delta_g \}$ in $\mathbb{N}_{>0}$ and
whose output is another one $\Delta'
=\{\delta'_0,\delta'_1,\ldots,\delta'_{g+1} \}$ such that
$\delta_i/\delta'_i=\delta_j/\delta'_j$ for $0\leq i,j\leq g$. From
Formula (\ref{z}) and Proposition \ref{siete}, it is clear that this
implies that the dual graph of the germ of curve $C_{\Delta'}$ that
Proposition \ref{siete} associates with $ \Delta'$ consists of the
one of $C_{\Delta}$ plus a new subgraph $\Gamma_{g+1}$. The
algorithm has the following steps:

\begin{itemize}
\item[1. ] Choose a natural number $z\geq 2$. \item[2. ] If there
exists $\delta'_{g+1}$ such that $\delta'_{g+1} \in \langle
\delta_0,\delta_1,\ldots,\delta_g \rangle$, $\delta'_{g+1} < z^2
\delta_g$ and $\gcd(z,\delta'_{g+1})=1$, then return
$\Delta'=\{z\delta_0,z\delta_1,\ldots, z\delta_g,
\delta'_{g+1}\}$. Else, \item[3. ] increase the value of $z$ and
go to Step 2.
\end{itemize}
This procedure ends because the semigroup $\langle
\delta_0,\delta_1,\ldots,\delta_g \rangle$ has a conductor.


\subsubsection{Type C weight functions.}
 Definition \ref{buena} shows  how to obtain  $\delta$-sequences in
 $\mathbb{Z}^{2}$ (and so type C weight functions). Next, we
 describe a particularly simple case of this type.
\vspace{1mm}

\noindent {\it Type C functions coming from AMS curves.} This type
of weight functions is simple because the easiness for obtaining
$\delta$-sequences attached to AMS curves. Indeed, fix a finite set
$\{n_i\}_{i=1}^g$ of positive integers, $n_i \geq 2$, $n_g>2$. We
know that the set $\Delta^* := \{\delta^*_i :=n_{i+1} n_{i+2} \cdots
n_g\}_{i=0}^{g-1} \cup \{1\}$ is a $\delta$-sequence in
$\mathbb{N}_{>0}$ associated with an AMS curve. Therefore, setting
\[
\begin{array}{ccc}
    \delta_0 & = & (n_{1} n_{2} \cdots n_{g-1} , n_{1} n_{2} \cdots n_{g-1})  \\
    \delta_1& = & (n_{2} n_{3} \cdots n_{g-1} , n_{2} n_{3} \cdots n_{g-1})  \\
    \vdots  & \vdots & \vdots \\
    \delta_{g-1} & = & (1, 1)  \\
    \delta_g & = & (0,1) \\
  \end{array}
\]
we get a $\delta$-sequence that provides a type C weight function.
The corresponding infinite dual graph $\Gamma$  coincides with the
one of the  AMS curve replacing, as usual, the last set of
proximate points  by infinitely many ones. Notice that the
$\delta_i$'s, $0 \leq i \leq g-1$, are in the line $ x -y =0$.

\subsubsection{Type D weight functions.}

\label{433} We are going to explicitly construct
$\delta$-sequences suitable for weight functions of type D. We
start with a $\delta$-sequence in $\mathbb{N}_{>0}$,
$\Delta=\{\delta_0,\ldots,\delta_{g-1}\}$,  $g \geq 2$, and a
positive non-rational number $a \in \mathbb{R}$ such that
\begin{equation} \label{eq}
a < n_{g-1} \delta_{g-1},
\end{equation}
$n_{g-1}$ being the usual number for the $\delta$-sequence
$\Delta$. Set $\langle a_1; a_2, a_3,\ldots, a_{j}, \ldots
\rangle$ the infinite continued fraction expansion given by $a$
and, for any index $j \geq 2$, let $m_{g-1}^{j}, e_{g-1}^{j}$ be
the relatively prime positive integers such that
$$\frac{m_{g-1}^{j}}{e_{g-1}^{j}}=\langle a_1; a_2, a_3,\ldots, a_{j}\rangle.$$

Next, we shall define a family of finite sequences $\left\{\Delta_j=
\{\delta_0^{j},\ldots\delta_{g-1}^{j},\delta_g^{j}\} \right
\}_{j=2}^{\infty}$ which, under suitable conditions, will give the
normalized $\delta$-sequences in $\mathbb{N}_{>0}$ providing our
$\delta$-sequence in $\mathbb{R}$. To do it, define
$\bar{\delta}_i^{j}:=e_{g-1}^{j} \delta_i$ ($i=0,1,\ldots,g-1$),
$\bar{\delta}_g^{j}:=n_{g-1}e_{g-1}^{j}\delta_{g-1}-m_{g-1}^{j}$ and
$\delta_i^{j}:=\bar{\delta}_i^{j}/\gcd(\bar{\delta}_i^{j}\mid
i=0,1,\ldots,g)$, $0 \leq i \leq g$.

Notice that, for $j \geq 4$, the equality
\begin{equation}
\label{lados}
\bar{\delta}_{g}^{j}=\bar{\delta}_{g}^{j-1}a_{j}+\bar{\delta}_{g}^{j-2}
\end{equation}
holds.  Indeed, $m_{g-1}^{j}=a_{j}m_{g-1}^{j-1}+m_{g-1}^{j-2}$ and
$e_{g-1}^{j}=a_{j}e_{g-1}^{j-1}+e_{g-1}^{j-2}$ \cite{ni-zu}.
Therefore
$$\bar{\delta}_g^{j}=n_{g-1}\left( a_{j}e_{g-1}^{j-1}+e_{g-1}^{j-2}\right)\delta_{g-1}-\left(a_{j}
m_{g-1}^{j-1}+m_{g-1}^{j-2}\right)=$$
$$\left(n_{g-1}e_{g-1}^{j-1}\delta_{g-1}-m_{g-1}^{j-1}\right)a_{j}+ \left( n_{g-1}e_{g-1}^{j-2}
\delta_{g-1}-m_{g-1}^{j-2}\right)=
\bar{\delta}_{g}^{j-1}a_{j}+\bar{\delta}_{g}^{j-2}, $$ as stated.
Soon, we shall give conditions in order to $\Delta_j$ be a
$\delta$-sequence. Firstly assume that there exists  a positive
integer $s_0$ such that $n_{g-1}\delta_{g-1}> \langle a_1; a_2,
a_3,\ldots, a_{j}\rangle $ for $j\in \{s_0, s_0+1\}$. Then, the
chain of equalities
$$\frac{\bar{\delta}_g^{j}}{e_{g-1}^{j}}=n_{g-1}\delta_{g-1}-\frac{m_{g-1}^{j}}{e_{g-1}^{j}}=
n_{g-1}\delta_{g-1}-\langle a_1; a_2, a_3,\ldots, a_{j} \rangle$$
and (\ref{lados}) prove that $\bar{\delta}_g^{j} >0$ for $j \geq
s_0$. We notice that the definition of $\bar{\delta}_g^{j}$ and
(\ref{eq}) prove that $s_0$ always exists.

Finally, we are ready to prove that $\Delta_j$ is a
$\delta$-sequence in $\mathbb{N}_{>0}$  for $j \geq s_1$ whenever
$s_1 \geq s_0$ is  an index such that the values
$\bar{\delta}_g^{s_1}$ and $\bar{\delta}_g^{s_1+1}$ belong to the
semigroup  in $\mathbb{N}_{>0}$ generated by
$\delta_0,\delta_1,\ldots,\delta_{g-1}$ (notice that $s_1$ exists
because this semigroup has a conductor). In fact, we only need to
show that $n_g^{j} \delta_g^{j}$ belongs to the semigroup
generated by $\delta_0^{j},\delta_1^{j},\ldots,\delta_{g-1}^{j}$
for all $j\geq s_1$, where $n_g^{j}$ denotes
$\gcd(\delta_0^{j},\ldots,\delta_{g-1}^{j})$. But this is
equivalent to the fact that $\bar{\delta}_g^{j}$ belongs to the
semigroup generated by $\delta_0,\delta_1,\ldots,\delta_{g-1}$.
Now, the result follows inductively from the hypothesis and
Equality  (\ref{lados}).

As a consequence, fixed a $\delta$-sequence in $\mathbb{N}_{>0}$, we
can determine non-rational real numbers $a$ such that all the above
sequences $\Delta_j$  are $\delta$-sequences in $\mathbb{N}_{>0}$,
giving rise after normalizing to a $\delta$-sequence in
$\mathbb{R}$. In fact, we must consider a positive integer $a_1
<n_{g-1}\delta_{g-1}-1$ and pick $a_2, a_3 \in \mathbb{N}_{>0}$ such
that both $e_{g-1}^{2}(n_{g-1}\delta_{g-1}- \langle a_1; a_2
\rangle)$ and $e_{g-1}^{3}(n_{g-1}\delta_{g-1}- \langle a_1; a_2,
a_3 \rangle)$ are in the semigroup of $ \mathbb{N}_{>0}$ generated
by $\delta_0,\delta_1,\ldots,\delta_{g-1}$. Then, it suffices to
take
$$a = a_1+\frac{1}{a_2+\frac{1}{a_3+\frac{1}{b}}},$$ where $b$ is
any non-rational positive real number. Thus, with the above
election the corresponding $\delta$-sequence in $\mathbb{R}$ is $
\left \{ \frac{\delta_0}{\delta_1},1,
\frac{\delta_2}{\delta_1},\ldots,\frac{\delta_{g-1}}{\delta_1},
\frac{1}{\delta_1}(n_{g-1}\delta_{g-1}-a) \right \}$.\\

We conclude this section by showing how to construct type E weight
functions and giving some examples.

\subsubsection{Type E weight functions.}
\label{434} A procedure for obtaining this type of weight
functions is, starting with a given $\delta$-sequence in
$\mathbb{N}_{>0}$, to reproduce indefinitely the procedure given
in Subsection \ref{ene} but with the following extra condition in
Step 2: $z\delta_g<\delta'_{g+1}$. This condition assures that,
normalizing the successively obtained $\delta$-sequences, one gets
an increasing $\delta$-sequence in $\mathbb{Q}$ (with the natural
ordering) and, then, it generates a well-ordered semigroup. For
instance, considering only values of $z$ which are relatively
prime with $\delta_g$ and taking $\delta'_{g+1}=(z+1)\delta_g$ in
Step 2, the above condition will be satisfied.






\subsubsection{Examples.}

\label{435} We begin this subsection by noting that the examples
of weight functions over the polynomial ring in two indeterminates
given by O'Sullivan in \cite{sul} are  of the types described in
this paper. Let us start  by Example 3.3 in \cite{sul}. There, up
to minor changes of notation, it is considered a pair $(r,s)$ of
relatively prime positive integers, $r < s$, and also the
corresponding integers $p$ and $q$ such that $pr-qs=1$, $0 < p <
s$. Setting $u'=y^s / x^r$ and $v'=x^q/y^p$, the author considers
the valuation given by $\nu(u')=(0,1)$ and $\nu(v')=(1,0)$ whose
valuation ring is $k[u',v',v'/u', v'/u'^2, \ldots]_{(u')}$. The
weight  function defined by $\nu$ has as order domain $k[x,y]$.
Expressed in our language, this is simply a type C weight function
given by the $\delta$-sequence $\{ \delta_0:=(s,p), \delta_1:=(r,
q) \}$. If we consider the continued fraction expansion of
$s/(s-r)= \langle a_1; a_2, \ldots, a_r \rangle$, it holds that
when we  apply the Euclidean algorithm to $\delta_0$ and $\delta_0
- \delta_1$ as described at the beginning of Section \ref{42}, we
get $r-1$ rows reproducing the case of $s$ and $s-r$ and,
afterwards, we cannot continue. So, the dual graph will have
infinitely many satellite points over the divisor obtained after
reproducing the blowing-up procedure given by the pair $(s, s-r)$.

Also, it is clear that Example 3.4 in \cite{sul} corresponds to a
type D weight function $w$ such that $w(y)=1$ and $w(x) = \tau$,
$\tau >1$ being an non-rational real number.

Finally, Example 5.2 in \cite{sul} corresponds to the simplest
type E weight function. The HNE of the associated valuation in a
regular system of parameters $\{u,v\}$ of the local ring
$\mathcal{O}_{\mathbb{P}^2,p}$ is

\begin{equation}
\begin{array}{lccl}\label{hne}
&v & = & u^{2} +u^{2}w_{1} \\
&u & = & w_1^{2}+ w_1^{2}w_{2} \\
&w_1 & = & w_2^{2}+ w_2^{2}w_{3} \\
&\vdots & \nonumber & \vdots \\
\end{array}
\end{equation}

\vspace{2mm} \noindent repeating the above structure indefinitely.
The first elements of a  family of $\delta$-sequences in
$\mathbb{N}_{>0}$ providing the desired $\delta$-sequence would be
$\{3,1\}, \{6,2,5\}, \{12,4,10,19\}, \ldots$ . \vspace{3mm}

Let us see other  examples. $\Delta =\{(10,10),(3,3),(24,25)\}$ is
a $\delta$-sequence in $\mathbb{Z}^2$ which comes from the
$\delta$-sequence in $\mathbb{N}_{>0}$ $\{40,12,97\}$. The
continued fraction $\langle a_1;a_2,a_3\rangle$ is $\langle
5;1,3\rangle$, since $m_1=23$ and $e_1=4$, so it holds that
$(A,B)=(1,1)$ and $(A',B')=(1,0)$.

Starting with the $\delta$-sequence in $\mathbb{N}_{>0}$
$\{11,9\}$, the development in Subsection \ref{433} shows that
$\Delta =\{11/9,1, (19- \frac{2 \sqrt{3}+1}{3 \sqrt{3} +1})/9\}$
is a $\delta$-sequence in $\mathbb{R}$. Indeed, the last element
in the $\delta$-sequence can be obtained by setting, with the
notation in that subsection, $a_1=80, a_2=1, a_3=2$ and $b =
\sqrt{3}$. $\Delta = \{12/8, 1,1/3, 3/4,$ $ 13/24, (11- \frac{
\sqrt{3}+1}{2 \sqrt{3} +2})/24\}$ is  another example of
$\delta$-sequence in $\mathbb{R}$. This is attached to the
$\delta$-sequence in $\mathbb{N}_{>0}$, $\{36,24,8,18,13\}$. Here,
$a_1=15, a_2=2, a_3=1$ and also, for convenience, $b = \sqrt{3}$.

\section{Evaluation codes given by $\delta$-sequences.}
\label{cinco}
\subsection{Generalized telescopic semigroups}\label{51}

A semigroup $S \subseteq \mathbb{N}_{\geq 0}$ is called  to be {\it
telescopic} if it is spanned by a finite sequence of positive
integers $\{\alpha_1, \alpha_2, \ldots, \alpha_r \}$ (named {\it
telescopic sequence}) such that
$ \gcd(\alpha_1, \alpha_2, \ldots, \alpha_r) =1$ and
$\alpha_i / \gcd(\alpha_1, \alpha_2, \ldots, \alpha_i)$ belongs to
the semigroup generated by the set $\{\alpha_1 / \gcd(\alpha_1,
\alpha_2, \ldots, \alpha_{i-1}), \ldots, \alpha_{i-1} /
\gcd(\alpha_1, \alpha_2, \ldots, \alpha_{i-1})\}$. An important
property of these semigroups is that each element $\alpha \in S$ can
be uniquely expressed in the form $\alpha = \sum_{i=1}^r a_i
\alpha_i$, provided that $a_1 \geq 0$ and $0 \leq a_i <
\gcd(\alpha_1, \alpha_2, \ldots, \alpha_{i-1})/ \gcd(\alpha_1,
\alpha_2, \ldots, \alpha_{i})$ ($2 \leq i \leq r$). This fact has
importance when one desires to bound the minimum distance of the
 dual codes of the evaluation ones given by classical weight functions \cite{h-l-p}.

The following definition is a natural enlargement in such a way
that   the mentioned property is preserved.

\begin{de}
\label{getel} {\rm A {\it generalized telescopic semigroup}  in
$\mathbb{Z}^2$ (respectively, $\mathbb{Q}$) (respectively,
$\mathbb{R}$) is a cancellative well-ordered commutative with zero
semigroup spanned by a set $A=\{\alpha_i\}_{i=1}^r$ such that:
\begin{description}
\item[($\mathbb{Z}^2$)] $r < \infty$, $A \subset \mathbb{Z}^2$,
which is lexicographically ordered,  the points in
$\{\alpha_i\}_{i=1}^{r-1}$ belong to the same line $L$ which passes
through $(0,0)$, $\alpha_r \not \in L$ and there exists a telescopic
sequence, $\{\beta_i\}_{i=1}^{r-1}$, such that the morphism of ordered
semigroups $\varrho: \langle \alpha_1, \ldots \alpha_{r-1} \rangle
\longrightarrow \langle \beta_1, \ldots \beta_{r-1} \rangle$ given
by $\varrho(\alpha_i) = \beta_i$ is an isomorphism.
\item[($\mathbb{Q}$)] (Respectively, $r = \infty$, $A \subset
\mathbb{Q}$ (natural order) and  for each $i>1$ there exists a
telescopic sequence, $\{\beta_{j}\}_{j=1}^{i}$, such that the
morphism of ordered semigroups $\rho : \langle \alpha_1, \ldots
\alpha_{i-1} \rangle \longrightarrow \langle \beta_{1}, \ldots
\beta_{i-1} \rangle$, $\rho(\alpha_j) = \beta_{j}$,  are
isomorphisms.)
\item[($\mathbb{R})$] (Respectively,  $r < \infty$, $A \subset
\mathbb{R}$ (natural order), $\alpha_i \in \mathbb{Q}, 0\leq i \leq
r-1$, $\alpha_r \in \mathbb{R} \setminus \mathbb{Q}$ and there
exists a telescopic sequence, $\{\beta_i\}_{i=1}^{r-1}$, such that
the  morphism of ordered semigroups $\rho$ defined as in ($\mathbb{Z}^2$) is an
isomorphism.)
\end{description}
}
\end{de}

Generically, a generalized telescopic semigroup in $\mathbb{Z}^2$,
$\mathbb{Q}$ or $\mathbb{R}$ will be named simply a generalized
telescopic semigroup.

\begin{prop}  \label{teles}
Let $S$ be a generalized telescopic semigroup spanned by $\{
\alpha_i\}_{i=1}^r$, $r \leq \infty$, as above. Then any element
$\alpha \in S$ can be written in a unique way of the form
\begin{equation}
\label{enteles} \alpha = \sum_{i=1}^s a_i \alpha_i,
\end{equation}
where $s<r+1$, $a_1, a_s \geq 0$  and $0 \leq a_i < \gcd(\alpha_1,
\alpha_2, \ldots, \alpha_{i-1})/ \gcd(\alpha_1, \alpha_2, \ldots,
\alpha_{i})$ ($2 \leq i < s$) and we have put
$\gcd(\alpha_1,\alpha_2, \ldots, \alpha_j) :=
\gcd(\beta_1,\beta_2, \ldots, \beta_j)$ ($1 \leq j < r$),
$\beta_j$ being as in Definition \ref{getel}.
\end{prop}
\begin{proof}
First let us see the case when $r$ is finite. Due to the nature of
the generators of $S$, when we set $\alpha = \sum_{i=1}^r a_i
\alpha_i$, it holds that the value $a_r$ must be unique, since
either it is  the unique nonnegative integer such that $\alpha - a_r
\alpha_r$ is on the line $L$ when $S \subset \mathbb{Z}^2$ or it is
the unique nonnegative integer such that $\alpha - a_r \alpha_r$ is
a rational value whenever $S \subset \mathbb{R}$. Now the semigroup
spanned by $\{ \alpha_i\}_{i=1}^{r-1}$ behaves as the one generated
by the elements, except the last one, of a telescopic semigroup.
Thus, if we set $a=\gcd(\alpha_1, \alpha_2, \ldots, \alpha_{r-1})$,
it happens that $\langle \alpha_1/a, \alpha_2/a, \ldots,
\alpha_{r-1}/a \rangle$ is like a telescopic semigroup and the fact
$$ (\alpha - a_r \alpha_r)/a \in \langle \alpha_1/a, \alpha_2/a,
\ldots, \alpha_{r-1}/a \rangle$$ provides the desired property.

When $r$ is infinite, the proof runs similarly. Indeed, assume that
$\alpha = \sum_{i=1}^s b_i \alpha_i$. If $a=\gcd(\alpha_1, \alpha_2,
\ldots, \alpha_{s})$, then $\langle \alpha_1/a, \alpha_2/a, \ldots,
\alpha_{s}/a \rangle$ behaves as a telescopic semigroup and,
therefore $\alpha/a = \sum_{i=1}^s a_i (\alpha_i/a)$, where the set
$\{a_i\}_{i=1}^s$ satisfies the desired properties and the result is
proved.
\end{proof}

\noindent {\it Remark}. A $\delta$-sequence in $\mathbb{N}_{>0}$
generates a telescopic semigroup. We have just enlarged this last
concept in such a way that any $\delta$-sequence $\Delta$ generates
a generalized telescopic semigroup and so $S_\Delta$ satisfies the
property given in the above proposition.

\subsection{Evaluation codes}
To construct error-correcting codes from a $\delta$-sequence, we
must consider the weight function $w_\Delta$ described in Theorem
\ref{main} and an epimorphism of $k$-algebras $ev:k[x,y]
\rightarrow k^n$, for some fixed positive integer $n$, which
usually will consist of evaluating $n$ previously picked points
$p_i$ $(1 \leq i \leq n)$ in $k^2$. The family of defined
evaluation codes will be $\{E_\alpha := ev(O_\alpha)\}_{\alpha \in
S_{\Delta}}$, $O_\alpha$ as in Theorem \ref{main}. The dual spaces
of the vector spaces $E_\alpha$ will be denoted by $C_\alpha$ and
they are the elements in the family of dual codes of the
evaluation ones.  Fixed $\alpha$ it suffices to compute the family
$\{ ev(\prod_{i=0}^m q_i^{\gamma_i})\}$, where $\prod_{i=0}^m
q_i^{\gamma_i}$ runs over the set of polynomials described at the
beginning of Subsection \ref{construction} for obtaining a
generator set of $E_\alpha$. Using linear algebra, it is easy to
compute bases for $E_\alpha$ and $C_\alpha$. Depending on $n$,
there is a positive integer $\Omega_n$ such that the vector spaces
$C_\alpha$ vanish if and only if $\alpha \geq \Omega_n$.

On the other hand, and as in the case of evaluation codes coming
from order functions on $\mathbb{N}_{\geq 0}$, denote by $\beta$
an  element in $S_\Delta$ and set
\[
\omega_\beta := \mbox{card}  \{ (\beta_1, \beta_2) \in
S^2_{\Delta} \; \mid \;\beta_1 + \beta_2 = \beta \}.
\]
The values
\[
d(\alpha) := \min \{ \omega_\beta | \alpha < \beta \in S_\Delta \}
\]
and
\[
d_{ev}(\alpha) := \min \{ \omega_\beta | \alpha < \beta \in
S_\Delta \mbox{ and $C_\beta \neq C_{\beta^{+}}$ } \},
\]
where $\beta^+ := \min \{ \gamma \in S_{\Delta} | \gamma > \beta
\}$ are named {\it Feng-Rao  distances} of $C_\alpha$. They
satisfy $d(C_\alpha) \geq d_{ev} (\alpha) \geq d(\alpha)$,
$d(C_\alpha)$ being the minimum distance of the code $C_\alpha$.
Above considerations and Proposition \ref{teles} prove the
following
\begin{theo}
\label{fr} Let $\Delta =\{\delta_i\}_{i=0}^r$, $r \leq \infty$, be
a $\delta$-sequence and $\{E_{\alpha}\}_{\alpha \in S_\Delta}$ and
$\{C_{\alpha}\}_{\alpha \in S_\Delta}$ the evaluation and dual
codes given by $\Delta$ and an epimorphism $ev$.
Then, the Feng-Rao distances  satisfy $d(\alpha) \leq \min [
\prod_{i=0}^s (a_i +1) ] -2 \leq d_{ev} (\alpha) $, where the
integer vectors $(a_0,a_1, \ldots, a_s)$ runs over the unique
coefficients of the corresponding expression (\ref{enteles}) for
$\{\delta_i\}_{i=0}^r$ instead of $\{\alpha_i\}_{i=1}^r$ of those
elements in $S_{\Delta}$ which are larger than or equal to $\alpha$
and smaller than $\Omega_n$.
\end{theo}

Dual evaluation codes given by classical weight functions admit
another lower bound of its minimal distance, called Goppa
 distance. In this case, the attached semigroup
to the weight function $S^*$ is numerical, its elements can be
enumerated according the natural ordering by a map $\chi:S^*
\rightarrow \mathbb{N}_{\geq 0}$ and, if $\alpha^* \in S^*$, the
Goppa bound of the  dual evaluation code associated with
$\alpha^*$ is $d_G (\alpha^*)= \chi(\alpha^*)+1- \xi_{S^*}$,
$\xi_{S^*}$ being the number of gaps of $S^*$. Let us see that a
similar bound can be given in our case. Let $\Delta$ be a
$\delta$-sequence. When $\Delta$ is finite, for our purposes of
ordering $S_\Delta$, we can always write $\Delta = \{\delta_0,
\delta_1, \ldots, \delta_{g-1}\} \cup \{\delta_g\}$, where the
ordered semigroup spanned by $\{\delta_0, \delta_1, \ldots,
\delta_{g-1}\}$ is isomorphic to a telescopic one that we shall
write $ S^* =\langle \delta_0^*, \delta_1^*, \ldots,
\delta_{g-1}^*\rangle$.  When $\Delta$ is not finite, to each
value $i \geq 0$, we shall associate a value $\xi_\Delta (i)$,
which will be the number of gaps of the semigroup generated by the
$\delta$-sequence in $\mathbb{N}_{>0}$, $\{\delta_0^*, \delta_1^*,
\ldots, \delta_{i}^*\}$, associated with the normalized one
$\{\delta_0, \delta_1, \ldots, \delta_{i}\}$.

Now, let $\alpha \in S_\Delta$. If $\Delta$ is finite, set $\alpha
= \sum_{i=0}^g a_i \delta_i$ the unique expression of $\alpha$ as
given in (\ref{enteles}). Set $B$ the least positive integer such
that $B \delta_g > \alpha$ and for $0 \leq j \leq B$, write $$A_j
+j \delta_g := \min \{z + j \delta_g | z +j \delta_g > \alpha
\mbox{\; and $ z \in \langle \delta_0, \delta_1, \ldots,
\delta_{g-1} \rangle$} \}.$$ If $A_j = \sum_{i=0}^{g-1} a^j_i
\delta_i$ according (\ref{enteles}), we denote $A_j^* =
\sum_{i=0}^{g-1} a^j_i \delta_i^*$ and define the {\it Goppa
distance of} $C_\alpha$ as $d_{S_\Delta} (\alpha) = \min \{ d_G
(A_j^*) (j+1) | 0 \leq j \leq B\}$.

For the infinite case, assuming as above $\alpha = \sum_{i=0}^s
a_i \delta_i$, set $\alpha^* = \sum_{i=0}^s a_i \delta_i^*$ and
the {\it Goppa  distance} of $C_\alpha$ will be
\[
d_{S_\Delta} (\alpha) := \chi(\alpha^*) +1 - \xi_{\Delta} (s).
\]

Taking into account properties of the classical Goppa distance and
of the  telescopic semigroups, it holds the following

\begin{prop}
\label{go} Let $\Delta =\{\delta_i\}_{i=0}^r$, $r \leq \infty$, be a
$\delta$-sequence and $\{E_{\alpha}\}_{\alpha \in S_\Delta}$ and
$\{C_{\alpha}\}_{\alpha \in S_\Delta}$ the corresponding evaluation
and dual codes for some fixed evaluation morphism $ev$. Then,
$d(\alpha) \geq d_{S_\Delta} (\alpha)$ and the Goppa distance of
$C_\alpha$ is
\[
d_{S_\Delta} (\alpha) = \left \{ \begin{array}{ll}  \min_{0 \leq j
\leq B} \left\{ [\chi (A_j^*) +1 -\frac{1+ \sum_{i=0}^{g-1} (n_i
-1) \delta^*_i}{2}]\; (j+1) \right\} & \mbox{if $s=g$ is
finite}\\
\chi(\alpha^*) +1 - \left (1+ \sum_{i=0}^{s} (n_i -1) \delta_i^*
\right ) /2 & \mbox{otherwise,}
\end{array}
\right.
\]
where $B$ is as above, $n_0=1$ and the remaining $n_i$ are the
usual ones for $\Delta$.
\end{prop}
\begin{proof}
First part can be proved by taking into account that when one
considers weight functions with values in a semigroup $S^*$ in
$\mathbb{N}_{\geq 0}$ spanned by values whose greatest common
divisor is one and with $g$ gaps, then the Feng-Rao distance
$d(\alpha)$ of the dual evaluation code associated with $\alpha \in
S^*$ is larger than or equal to its Goppa distance. Last part is a
consequence of the computation given in \cite{h-l-p} of the number
of gaps of a telescopic semigroup.
\end{proof}

Finally, we prove that the semigroups $S_\Delta$ corresponding to
$\delta$-sequences $\Delta$ providing weight functions of type C are
simplicial. This one will be a useful property since there exists an
algorithm  given by Ruano in \cite{rua} for computing the Feng-Rao
distance $d(\alpha)$ of codes $C_\alpha$ associated with order
functions with simplicial image semigroup included in
$(\mathbb{N}_{\geq 0})^r$, $r\geq 1$. Recall what this concept
means. Let $S$ be a semigroup included in $(\mathbb{N}_{\geq 0})^r$,
$r\geq 1$. Setting $U$ an indeterminate, the $k$-algebra $k[S]:=
\bigoplus_{\alpha \in S} k U^\alpha$, where the product of
polynomials is induced by $(a U^\alpha)  (b U^\beta) = (ab)
U^{\alpha + \beta}$, $a,b \in k$ and $\alpha, \beta \in S$ is named
the {\it $k$-algebra of the semigroup $S$}. Now, set $C_S$ the cone
in $\mathbb{R}^r$ spanned by $S$. $C_S$ is a strongly convex cone.
We shall say that $S$ is {\it simplicial} whenever the dimension of
the $k$-algebra $k[S]$ coincides with the number of extremal rays of
the cone $C_S$.

\begin{prop}
The semigroup $S_\Delta$ of a $\delta$-sequence  $ \Delta = \{
\delta_i \}_{i=0}^g \subseteq (\mathbb{N}_{\geq 0})^2 $  of type C
is simplicial.
\end{prop}
\begin{proof}
Since $\Delta$ is of type  C, the number of extremal rays of
$C_{S_\Delta}$ is two. Now, set $k[V_0,V_1, \ldots,V_g]$ the
polynomial ring in $g+1$ indeterminates and set
$$\psi:k[V_0,V_1, \ldots,V_g] \rightarrow k[S_\Delta]$$ the morphism
of $k$-algebras given by $\psi(V_i)=U^{\delta_i}$. $\psi$ is an
epimorphism. For $0 < i < g$, consider the unique expression
$$
n_i \delta_i = \sum_{j=0}^{i-1} a_{ij} \delta_j,
$$
such that $a_{i0} >0$ and $0 \leq a_{ij} < n_j$ for the remaining
indices $j$, being $n_j$ the usual number associated with
$\Delta$. Then, the proof follows by taking into account that the
kernel of $\psi$, $I$, is the ideal of $k[V_0,V_1, \ldots,V_g]$
spanned by the set $G := \{ V_i^{n_i} - \prod_{j=0}^{i-1}
V_j^{a_{ij}} \}_{0 <i <g}$ and therefore $k[S] \cong k[V_0,V_1,
\ldots,V_g] / I$ whose dimension is also two.

Only remains to prove that $G$ spans $I$. To do it, it suffices to
notice that $I$ is generated by the binomials $A -B \in k[V_0,V_1,
\ldots,V_g]$ such that $A$ and $B$ are homogeneous monomials with
coefficient 1 and $\psi (A-B)=0$. These monomials can be represented
by the set $\mathcal{B}$ of pairs $(a,b) \in (\mathbb{N}_{\geq
0}^{g+1})^2$ representing the exponents of both monomials.
$\mathcal{B}$ gives rise to a congruence, that is an equivalence
binary relation such that if $c \in \mathbb{N}_{\geq 0}^{g+1}$ and
$(a,b) \in \mathcal{B}$, then $(a+c,b+c) \in \mathcal{B}$. As the
set $\mathcal{D}$ of pairs given by the exponents of the elements in
$G$, satisfies that $\mathcal{B}$ corresponds to the smallest
congruence containing $\mathcal{D}$, we get that $G$ spans the ideal
$I$, which concludes the proof (see the proof of \cite[Theorem
5.2]{ga-sa} for a more detailed explanation of a close result).
\end{proof}

\subsection{Examples}

First of all, we prove that Reed-Solomon codes can be regarded as
particular cases of evaluation codes attached to type C weight
functions.
\begin{prop}
\label{rs} Consider the finite field $k=\mathbb{F}_q$, a
$\delta$-sequence in $\mathbb{Z}^2$, $\Delta= \{
\delta_0=(p_1,p_2), \delta_1=(q_1,q_2)\}$, and the epimorphism
$ev$ given by evaluating $d \leq q-1$ points in a line in
$\mathbb{F}_q^2$. Then the family of evaluation codes
$\{E_\alpha\}_{\alpha \in S_\Delta}$ is the family of length $d$
Reed-Solomon codes associated with $\mathbb{F}_q$.
\end{prop}
\begin{proof}
Assume that we evaluate points at the line given by $y -ax -b =0$;
$a,b \in \mathbb{F}_q$. Our approximates are $x$ and $y$ and, in
order to span the spaces $O_\alpha$, we must use monomials in $x$
and $y$. When we evaluate a monomial $x^m y^n$, this corresponds to
evaluate $x^m(ax+b)^n$, that is  a polynomial of degree $m+n$ in the
indeterminate $x$. $p_1  \geq q_1$ and, since $\delta_1 < 2 \delta_1
< \cdots < l \delta_1$ are elements in $S_\Delta$, when we consider
$E_{l \delta_1}$, $l \in \mathbb{N}_{> 0}$, one evaluates among
others a monomial of degree $l$ in $x$. So, it suffices to show that
no monomial  of degree larger than $l$ appears in $O_{l \delta_1}$.
That is, we must prove that $ r \delta_0 + s \delta_1 < t \delta_1$,
$r, s \in \mathbb{N}_{\geq 0}$, implies $r + s \leq t$ and this is
true because then $r p_1 + s q_1 \leq t q_1$ and thus $(s-t)q_1 \leq
- r p_1 \leq - r q_1$, which concludes the proof when $q_1 >0$.
Otherwise the proof follows because then $r=0$.
\end{proof}

To finish this paper, we show some parameters of several families
of dual codes $C_\alpha$ attached to $\delta$-sequences. Our data
are computed using the computer algebra system {\sc Singular}
\cite{sin}.

\begin{example}
\label{c1}  {\rm Fix the field $\mathbb{F}_7$ and the
$\delta$-sequences $\Delta_1$, $\Delta_2$ and $\Delta_3$ of types C,
D and E respectively given, as we have described, by the sequence in
$\mathbb{N}_{>0}$, $\{11,9\}$. $\Delta_1 = \{(5,1),(4,1)\}$. Indeed,
with the notation in Definition \ref{buena}, $\delta^*_0=11,
\delta^*_1=9$, $t=2, a_1=5, a_2=2$. So
$\delta_0=\underline{y}_1=(5,1)$ and $\delta_1= \delta_0 -
\underline{y}_0 =(4,1)$. Our election for $\Delta_2$ is $\Delta_2 =
\{11/9, 1, (19- \frac{2 \sqrt{3}+1}{3 \sqrt{3} +1})/9 \thickapprox
2,031105\}$ (see Subsection \ref{435}). $\Delta_3$, and any
$\delta$-sequence of type E given in the sequel, is constructed as
we described at the end of Subsection \ref{434}, using the least
value $z$ we can choose. In this case the first four elements of
$\Delta_3$ are $11/9,1,3/2,9/4$.

Consider the map $ev$ given by evaluating at the points in
$\mathbb{F}_7^2$:
\[ \{ (1,1), (2,2), \ldots, (6,6), (1,2), (1,3), \ldots, (1,6),
(2,1)\}.
\]
$q_0=x, q_1=y$ are approximates for the corresponding to $\Delta_1$
case and $q_0, q_1$ and $q_2= y^{11} - x^9$ otherwise. Tables
\ref{doss} and \ref{dossbis} show parameters (dimension $k$, minimum
distance $d(C_\alpha)$ and Feng-Rao distance $d_{ev}(\alpha)$) of
the successive codes $C_\alpha$. Furthermore, only for $\Delta_1$,
we add the exponents of the approximates we use to get the new
generator to add to the previous ones for obtaining a basis of
$O_\alpha$. That is $O_{(4,1)}$ is generated by $1$ and $q_1$,
$O_{(5,1)}$ by $1$, $q_1$ and $q_0$, $O_{(8,2)}$ by $1$, $q_1$,
$q_0$ and $q_1^2$, and so on.
\begin{table}[h,t,b]
\caption[]{First Case in Examples \ref{c1}} \label{doss}
\begin{tabular}{ccccc}
  \hline
  $\alpha$ & exp & $k$ & $d_{\Delta_1} (C_\alpha) $  & $d_{ev, \Delta_1} (\alpha)$\\
  \hline
   (4,1)& 01 &10 & 2 & 2   \\
  (5,1)& 10 &9 & 3 & 3 \\
 (8,2)& 02 & 8 & 4 & 3 \\
   (9,2)& 11 &7 & 4 & 3\\
  (10,2)& 20 &6  & 4 & 4\\
   (12,3)& 03 &5 & 5 & 5 \\
 $^*$(13,3)& 12 & 4 & 5 & 5\\
   (16,4)& 04 &3 & 6 & 6\\
  $^{**}$(17,4)& 13 &2 & 6 & 6 \\
 (20,5)& 05 &1 & 10 & 10\\
  \hline
\end{tabular}
\end{table}

\begin{table}[h,t,b]
\caption[]{First Case in Examples \ref{c1}} \label{dossbis}
\begin{tabular}{ccccc}
  \hline
  $k$ &
  $d_{\Delta_2} (C_{\alpha'})$ &
 $d_{ev,\Delta_2} (\alpha')$ &  $d_{\Delta_3}(C_{\alpha"})$ & $d_{ev,\Delta_3}  (\alpha")$\\
  \hline
   10 & 2 & 2  & 2 & 2  \\
  9 & 3 & 2 &  3 & 2\\
 8 & 4 & 2 &  4 & 2 \\
  7 & 4 & 3 &  4 & 2 \\
  6  & 4 & 3 &  4 & 2 \\
  5 & 4 & 4 &   4 & 4\\
 4 & 6 & 4 &  5 & 4  \\
  3 & 6 & 4 &  5 & 5  \\
  2 & 6 & 4 &  7 & 4  \\
 1 & 6 & 5 &  6 & 7  \\
  \hline
\end{tabular}
\end{table}

Notice that the length of the code is $n=12$ and that the classical
parameters of the codes in the first three rows cannot be improved
(see \cite[Theorem 5.3.10]{rom}). Also, if we set $(14,3)$ $2 \; 1$
or $(15,3)$ $3 \; 0$ as $\alpha$ and the exponents instead of those
given in $^*$ either $(18,4)$ $2 \; 2$; $(19,4)$ $3 \; 1$ or
$(20,4)$ $4 \; 0$ instead of the values given in $^{**}$ we get the
same parameters for the corresponding codes
attached to $\Delta_1$.\\

Let us see another case, where we have used the $\delta$-sequence in
$\mathbb{N}_{>0}$, $\{36,24,8,18,13\}$. Here $\Delta_1 =
\{(18,0),(12,0), (4,0), (9,0), (7,-1)\}$ and our election for
$\Delta_2$ is $\{12/8, 1,1/3, 3/4,$ $ 13/24, (11- \frac{
\sqrt{3}+1}{2 \sqrt{3} +2})/24 \thickapprox 0,441666\}$. The
approximates for $\Delta_1$ and  $\Delta_3$ are $q_0=x, q_1=y,$
$q_2=-x^2+y^3$, $q_3= -x^6+3x^4y^3-3x^2y^6+y^9-y$ and $q_4=
x^{12}+x^{10}y^3+x^8y^6+x^6y^9+2x^6y+x^4y^{12}+x^4y^4+x^2y^{15}-x^2y^7-x+y^{18}-2y^{10}+y^2$
and we must add another large polynomial $q_5$ for  $\Delta_2$. The
table with the same parameters as above is displayed in Table
\ref{tress}.

\begin{table}[h,b,t]
\caption[]{Second Case in Examples \ref{c1}}\label{tress}
\begin{tabular}{ccccccccc}
  \hline
 $k$ & $d_{\Delta_1}(C)$ & $d_{ev, \Delta_1}$ &  $d_{\Delta_2}(C)$ &
 $d_{ev,\Delta_2}$  & $d_{\Delta_3} (C)$ & $d_{ev,\Delta_3} $\\
  \hline
   10 & 2 & 2  & 2 & 2 & 2 & 2  \\
   9 & 3 & 2  & 2 & 2 &  2 & 2\\
 8 & 3 & 2  & 3 & 2 &  2 & 2 \\
   7 & 4 & 3  & 4 & 2 &  4 & 4 \\
  6  & 4 & 3   & 5 & 3 &  4 & 4 \\
   5 & 4 & 3  & 5 & 3 &   4 & 4\\
  4 & 4 & 4 &  6 & 3 &  4 & 4  \\
   3 & 7 & 4  & 7 & 4 &  4 & 4  \\
  2 & 7 & 4  & 9 & 4 &  8 & 5  \\
 1 & 11 & 5  & 12 & 5 &  8 & 8  \\
  \hline
\end{tabular}
\end{table}
}
\end{example}

\begin{example}
\label{c2} {\rm Next, we show another examples. First, we use only
type C weight functions and, afterwards, we use weight functions of
all described types. Consider the finite field $\mathbb{F}_{2^5}$
and $\xi$ a primitive element. Tables \ref{cuatross} and
\ref{cuatrossbis} show the same  parameters as above corresponding
to some dual codes associated with the evaluation at the following
31 points in $\mathbb{F}_{2^5}^2$:
\[ \{ (\xi,\xi), (\xi,\xi^2), \ldots, (\xi,\xi^{14}), (\xi^2,\xi),
(\xi^2,\xi^2), \ldots, (\xi^2,\xi^{14}),
(\xi^3,\xi^3),(\xi^4,\xi^4),(\xi^5,\xi^5)\},
\]
and with the $\delta$-sequences of type C $\Delta_1 = \{(21,0),
(15,0), (35,0), (39,-1)\}$, $\Delta_2 =\{(2,1),(1,1)\}$ and
$\Delta_3=\{(5,5),(2,2), (7,8)\}$. Only for $\Delta_1$, we add the
values $\alpha$ and the exponents for the approximates which are
$q_0=x, q_1=y, q_2= y^7 + x^5$ and
$q_3=x^{15}+x^{10}y^7+x^{15}y^{14}+x^5+y^{21}$. $k$ remains valid for the codes of each row.\\

\begin{table}[h,t,b]
\caption[]{First Case in Examples \ref{c2}} \label{cuatross}
\begin{tabular}{ccccc}
  \hline
  $\alpha$ & exp & $k$ & $d_{\Delta_1} (C_\alpha)$  & $d_{ev,\Delta_1} (\alpha)$\\
  \hline
  (15,0) & 0100 & 29 & 2 & 2  \\
  (21,0)& 1000 & 28& 3 & 2   \\
  (30,0) & 0200 & 27 & 3 & 2  \\
  (35,0) & 0010 & 26 & 4 & 2  \\
  (36,0) & 1100 & 25 & 4 & 2  \\
  (39,-1) & 0001 & 24 & 4 & 3  \\
  (42,0) & 2000 & 23 & 4 & 3  \\
  (45,0) & 0300 & 22 & 5 & 3  \\
  (50,0) & 0110 & 21 & 5 & 3  \\
  \hline
\end{tabular}
\end{table}

\begin{table}[h,t,b]
\caption[]{First Case in Examples \ref{c2}} \label{cuatrossbis}
\begin{tabular}{ccccc}
  \hline
    $k$ & $d_{\Delta_2} (C_{\alpha'})$ & $d_{ev,\Delta_2} (\alpha')$ &
  $d_{\Delta_3} (C_{\alpha"})$ &  $d_{ev,\Delta_3} (\alpha")$ \\
  \hline
  29 & 2 & 2 & 2 & 2  \\
  28 & 3 & 3 & 2 & 2  \\
  27 & 3 & 3 & 3& 2  \\
  26 & 3 & 3 & 3 & 2 \\
  25 & 4 & 4 & 3 & 2 \\
  24 & 5 & 5 & 5 & 3  \\
  23 & 5 & 5 & 6 & 3 \\
  22 & 5 & 5 & 6& 3  \\
  21 & 6 & 6 & 7 & 3 \\
  \hline
\end{tabular}
\end{table}

Now, we consider the same set of evaluation points, but $\Delta_1$,
$\Delta_2$ and $\Delta_3$ are $\delta$-sequences of types C, D and E
related with the $\delta$-sequence in $\mathbb{N}_{>0}$,
$\{36,24,8,18,13\}$. Concretely, $\Delta_1 = \{(18,0),(12,0), (4,0),
(9,0), (7,-1)\}$, $\Delta_3$ is as we described at the end of
Subsection \ref{434} and we take as election for $\Delta_2$,
$\Delta_2= \{12/8, 1,1/3, 3/4, 13/24, (6- \frac{ 2 \sqrt{3}+1}{11
\sqrt{3} +5})/24 \thickapprox 0,242266\}$. A partial table with
parameters as above is displayed in Table \ref{cincos}.

\begin{table}[h,b,t]
\caption[]{Second Case in Examples \ref{c2}} \label{cincos}
\begin{tabular}{ccccccc}
  \hline
 $k$ & $d_{\Delta_1}(C)$ & $d_{ev, \Delta_1}$  & $d_{\Delta_2}(C)$ &
 $d_{ev,\Delta_2}$  & $d_{\Delta_3}(C)$ & $d_{ev,\Delta_3} $\\
  \hline
   29 & 2 & 2  & 2 & 2 & 2 & 2  \\
   28 & 3 & 2  & 3 & 2  & 3 & 2\\
    27 & 3 & 2  & 3 & 2  & 3 & 2 \\
   26 & 4 & 3  & 3 & 2  & 4 & 2 \\
  25 & 5 & 3   & 3 & 2   & 5 & 2 \\
   24 & 5 & 3  & 4 & 2  & 5 & 2\\
  23 & 5 & 3 &  6 & 2  & 5 & 2  \\
   22 & 7 & 3  & 6 & 4  & 7 & 2  \\
  21& 7 & 3 &  7 & 4  & 7 & 2  \\
  \hline
\end{tabular}
\end{table}
}
\end{example}

\begin{example}
\label{c3} {\rm  Finally, we consider the same field of Examples
\ref{c2} and the $\delta$-sequences $\Delta_1= \{(3,1),(2,1)\}$,
$\Delta_2= \{7/5,1, (7- \frac{ \sqrt{3}+1}{4 \sqrt{3} +3})/5
\thickapprox 1,34496\}$ (for convenience our examples of type D
valuations always use $b= \sqrt{3}$ according the notation in
Subsection \ref{433}) and $\Delta_3$ of type E related with the
$\delta$-sequence in $\mathbb{N}_{>0}$, $\{7,5\}$. The family of
points to evaluate is
\[
\{ (\xi,\xi), (\xi,\xi^2), \ldots, (\xi,\xi^{14}), (\xi^6,\xi),
(\xi^6,\xi^2), \ldots, (\xi^6,\xi^{10}),(\xi^2,\xi^{11}),
(\xi^2,\xi^{12}), (\xi^2,\xi^{13}),
\]
\[
(\xi^2,\xi^{14}),(\xi^{20},\xi^{20}), (\xi^{21},\xi^{21}),
(\xi^{28},\xi^{28}) \}. \] The (partial) corresponding  table is
given in Table \ref{seiss}.

\begin{table}[h,t,b]
\caption{Examples \ref{c3}} \label{seiss}
\begin{tabular}{ccccccc}
  \hline
 $k$ & $d_{\Delta_1} (C)$ & $d_{ev, \Delta_1}$ & $d_{\Delta_2} (C)$ &
 $d_{ev,\Delta_2}$  & $d_{\Delta_3} (C)$ & $d_{ev,\Delta_3} $\\
  \hline
  29 & 2 & 2   & 2 & 2 & 2 & 2\\
  28 & 3 & 3  & 3 & 2 & 3 & 2 \\
  27 & 4 & 3   & 3& 3  & 3 & 2 \\
  26 & 4 & 3  & 4 & 3 & 4 & 2 \\
  25 & 4 & 4   & 5 & 3  & 4 & 2 \\
  24 & 5 & 4   & 5 & 3  & 4 & 2  \\
  23 & 5 & 4   & 6 & 3  & 5 & 2 \\
  22 & 5 & 4   & 6 & 3  & 5 & 2 \\
  21 & 6 & 4   & 6 & 4  & 5 & 2 \\
\hline
\end{tabular}
\end{table}
}
\end{example}

\end{document}